\documentclass[journal,letterpaper,twocolumn,twoside,nofonttune]{IEEEtran}
\usepackage{etex}
\normalsize
\usepackage[T1]{fontenc}
\usepackage{amsmath,amssymb,amsfonts}
\usepackage{mathrsfs}
\usepackage{mathabx}
\usepackage{amsbsy}
\usepackage{graphicx}
\usepackage{graphics}
\usepackage{epstopdf}
\usepackage{algorithm}
\usepackage{algorithmic}
\usepackage{subfigure}
\usepackage{cite}
\usepackage{multirow}
\usepackage{ctable}
\usepackage{caption}

\title{\Huge$\,$\\[-2.75ex]
{Polar Coding for Non-Stationary Channels}\\[0.50ex]}

\author{\large%
Hessam Mahdavifar,\,\,\IEEEmembership{Member,~IEEE}
\vspace{-.25in}
\thanks{%
 The material in this paper was presented in part at the IEEE International Symposium on Information Theory in June 2017.
}
\thanks{%
This work was supported in part by the National Science Foundation,
under grant CCF--1763348.
}
\thanks{H.\ Mahdavifar is with the Department of Electrical Engineering and Computer Science, University of Michigan, Ann Arbor, MI 48104 (email: hessam@umich.edu).}
}

\newtheorem{theorem}{{Theorem}}
\newtheorem{lemma}[theorem]{{Lemma}}

\newtheorem{corollary}[theorem]{{Corollary}}

\newtheorem{definition}{{Definition}}

\newcommand{\cC}{{\cal C}} 

\newcommand{\cE}{{\cal E}}

\newcommand{\cN}{{\cal N}}

\DeclareMathAlphabet{\mathbfsl}{OT1}{ppl}{b}{it} 

\def\QEDclosed{\mbox{\rule[0pt]{1.3ex}{1.3ex}}} 

\def\QED{\QEDclosed} 

\newcommand{\be}[1]{\begin{equation}\label{#1}}
\newcommand{\ee}{\end{equation}} 
\newcommand{\eq}[1]{(\ref{#1})}

\renewcommand{\leq}{\leqslant}
 
\renewcommand{\geq}{\geqslant}

\newcommand{\script}[1]{{\mathscr #1}}

\renewcommand{\Bbb}{\mathbb}
 
\newcommand{\N}{{\Bbb N}}
\newcommand{\R}{{\Bbb R}}

\newcommand{\Tref}[1]{Theo\-rem\,\ref{#1}}

\newcommand{\Lref}[1]{Lem\-ma\,\ref{#1}}
\newcommand{\Cref}[1]{Co\-ro\-lla\-ry\,\ref{#1}}

\newcommand{\deff}{\mbox{$\stackrel{\rm def}{=}$}}

\newcommand{\sY}{\script{Y}}

\newcommand{\shalf}{\mbox{\raisebox{.8mm}{\footnotesize $\scriptstyle 1$}
\footnotesize$\!\!\! / \!\!\!$ \raisebox{-.8mm}{\footnotesize
$\scriptstyle 2$}}}

\newcommand{\IN}{\overline{I}_N}

\begin{document}

\maketitle

\begin{abstract}

The problem of polar coding for an arbitrary sequence of independent binary-input memoryless symmetric (BMS) channels $\left\{W_i\right\}_{i=1}^{N}$ is considered. Such a sequence of channels is referred to as a \textit{non-stationary} sequence of channels and arises in applications where data symbols experience different and independent channel characteristics. The sequence of channels is assumed to be completely known to both the transmitter and the receiver (a coherent scenario). Also, at each code block transmission, each of the channels is used only once. In other words, a codeword of length $N$ is constructed and then the $i$-th encoded bit is transmitted over $W_i$. The goal is to operate at a rate $R$ close to the average of the symmetric capacities of $W_i$'s, denoted by $\IN$. To this end, we construct a polar coding scheme using Ar{\i}kan's channel polarization transform in combination with certain permutations at each polarization level and certain skipped operations. In particular, given a non-stationary sequence of BMS channels $\left\{W_i\right\}_{i=1}^{N}$ and $P_e$, where $0 < P_e <1$, we construct a polar code of length $N$ and rate $R$ guaranteeing a block error probability of at most $P_e$ for transmission over $\left\{W_i\right\}_{i=1}^{N}$ such that 
$$
N \leq \frac{\kappa}{(\IN - R)^{\mu}},
$$ 
where $\mu$ is a constant and $\kappa$ is a constant depending on $P_e$ and $\mu$. We further show a numerical upper bound on $\mu$ that is: $\mu \leq 7.34$ for non-stationary binary erasure channels and $\mu \leq 8.54$ for general non-stationary BMS channels. The encoding and decoding complexities of the constructed polar code preserve $O(N \log N)$ complexity of Ar{\i}kan's polar codes. In an asymptotic sense, when coded bits are transmitted over a non-stationary sequence of BMS channels $\left\{W_i\right\}_{i=1}^{\infty}$, our proposed scheme achieves the average symmetric capacity
$$
\overline{I}(\left\{W_i\right\}_{i=1}^{\infty}) \ \deff\ \lim_{N\rightarrow \infty} \frac{1}{N}\sum_{i=1}^N I(W_i),
$$
assuming that the limit exists. 

\end{abstract}

\begin{keywords} 
Polar codes, finite-length analysis, non-stationary channels, channel polarization
\end{keywords}

\section{Introduction} 
\label{sec:Introduction}

\noindent 
\PARstart{P}{olar} codes were introduced by Ar{\i}kan in the seminal work of \cite{Arikan}. Polar codes are the first family of codes for the class of binary-input symmetric discrete memoryless channels that are provable to be capacity-achieving with low encoding and decoding complexity~\cite{Arikan}. Polar codes and polarization phenomenon have been successfully applied to a wide range of problems including data compression~\cite{Arikan2,abbe2011polarization}, broadcast channels~\cite{mondelli2015achieving,goela2015polar}, multiple access channels~\cite{STY,MELK}, and physical layer security~\cite{MV,andersson2010nested,chou2015polar,wei2016polar}. The error exponent and finite-length scaling behavior of polar codes are also well-studied~\cite{AT,hassani,guruswami,MUH,GB,FV,pfister2016near,fazeli2017binary}. 

Since the invention of polar codes, a problem of significant interest has been the following: can polar codes achieve the capacity limit in more general point-to-point channel models than the original one considered by Ar{\i}kan? Ideally, one may want to find the \emph{most general} set-up in which polar codes can be constructed to achieve the point-to-point capacity limit. As stated before, polar codes were originally shown to achieve the capacity of binary-input symmetric discrete memoryless channels. Later, they were extended to channels with non-binary discrete inputs \cite{csacsouglu2009polarization, sahebi2013multilevel,gulcu2018construction}, asymmetric channels \cite{honda2013polar}, channels with memory \cite{csacsouglu2016polar,wang2015construction}, and continuous-input channels such as additive white Gaussian noise (AWGN) channel \cite{abbe2011polar,li2013practical} and fading channel \cite{si2014polar}. However, in all these extensions of polar codes, there is a common assumption that the polarization transform is applied to a stationary sequence of channels. Transmission over stationary channels is, in general, a common assumption in channel coding problems and hence, it is not often mentioned at all as it is considered the \emph{default} assumption. In this work, we aim at extending polar codes beyond the stationary set-up. More specifically, transmissions over a \textit{non-stationary} sequence of channels are considered, as explicitly defined next.

\subsection{Problem formulation and motivation}

Suppose that an arbitrary sequence of independent binary-input memoryless symmetric (BMS) channels $\left\{W_i\right\}_{i=1}^{N}$ is given. A coherent communication over these channels is assumed, i.e., the sequence of memoryless channels $\left\{W_i\right\}_{i=1}^{N}$ is assumed to be completely known to both the transmitter and the receiver. We refer to this sequence as a non-stationary sequence of BMS channels or simply non-stationary BMS channels ~\cite{Alsan}. As we only consider BMS channels in this paper, a shorter phrase of \textit{non-stationary channels} is used throughout the paper keeping in mind that all the channels under discussion are BMS. We consider the problem of channel coding, and in particular polar coding, for a given non-stationary sequence of channels. It is assumed that each channel is used once, i.e., the code block length is $N$, equal to the number of channels. This clarifies the rationale behind calling this a non-stationary sequence. In a \textit{standard} stationary setting, the sequence of channels over which coded bits are transmitted form a stationary process, i.e., the channel is either fixed through the course of one code block transmission or it changes according to some channel law that does not depend on the coded bit index. 

The goal here is to construct polar coding schemes with rates \emph{close} to the average of the symmetric capacities of $W_i$'s in a non-stationary sequence. More precisely, we study trade-offs between code block length, probability of error, and gap to the average capacity. This, especially for stationary settings, is a well-studied topic and is referred to as \textit{non-asymptotic} analysis or \textit{finite-length} analysis, interchangeably, in the literature. In particular, it is well-known that given a channel $W$ with capacity $C$ the minimum possible block-length $N$ required to achieve a rate $R$ with block error probability at most $P_e$ is roughly equal to \cite{dob,str,poly}
\be{min-length}
 N \approx \frac{V( Q^{-1}\bigl(P_{\rm e})\bigr)^2 }{(C-R)^2},
\ee
where $V$ is a characteristic of the channel referred to as channel dispersion in \cite{poly} and $Q$-function is defined as $Q(z)\,\deff\,\int_z^{\infty} \frac{e^{-x^2/2}}{\sqrt{2\pi x}} dx.$ In light of this result, we study the performance of our proposed polar coding schemes over non-stationary channels by upper bounding the gap to the average capacity in terms of the code block length $N$. Note that the terms ``capacity'' and ``symmetric capacity'' are used interchangeably provided that the channels under discussion are symmetric. This is because the symmetric capacity is equal to the capacity for symmetric channels. 

Besides the theoretical motivation to study polar coding beyond the original set-up considered by Ar{\i}kan \cite{Arikan} there are practical applications indicating non-stationary scenarios are becoming increasingly important in wireless systems. An application of non-stationary channels that one may immediately think of is transmission over a time-varying channel. While this is certainly a valid model, the assumption that both the transmitter and the receiver perfectly know the channel at each time instance is not realistic. Alternatively, we argue that the non-stationary scenarios become relevant in high data-rate applications where users enjoy a wide range of available frequency spectrum, spatial diversity, and high-order modulations. Next, we discuss this in the context of next generation of wireless networks, namely 5G. Note that this is for illustration purposes only and many details are skipped. The reader is referred to the mentioned references, and the references therein, for the detailed discussion on the subject.

It is expected that next generation 5G networks will deliver data rates up to $20$\,Gbps \cite{3gpp-5g}, an order of $20$ increase comparing to the legacy 4G networks. In order to meet such expectations, many new features have been introduced in the 5G standard. This includes communicating in millimeter-wave (mmWave) band together with enhanced carrier aggregation capabilities \cite{3gpp-5g-2} and \textit{massive} multiple-input multiple-output (MIMO) transmissions \cite{massive}. In 5G systems, aggregation of up to $16$ carriers will be deployed across the extensive available frequency spectrum including the bands below $6$\,GHz, the radio frequency (RF) band, as well as bands around $30$\,GHz, the mmWave band \cite{3gpp-5g-2}. This implies that transmissions may take place over up to $16$ different and \textit{almost} independent channels with different characteristics. This is mainly due to significantly different attenuations and diffractions that electromagnetic waves experience in different frequency bands. Another main feature of 5G systems is the massive MIMO transmission. It is expected that up to $64 \times 64$ MIMO transmissions will be enabled in 5G \cite{massive}. Under a coherent transmission scenario a $64 \times 64$ MIMO channel can be decomposed into $64$ parallel spatial channels which are independent \cite{zheng2003diversity}. Furthermore, a bit-interleaved coded modulation (BICM) together with quadrature amplitude modulation (QAM) are deployed in current cellular systems. Note that a $256$-QAM modulation, supported in current systems, together with a \emph{truly} random interleaver can be decomposed into $8$ independent parallel channels \cite{caire1998bit} (with a certain symbol labeling rule, there are actually $4$ pairs of channels and channels in each pair are independent but identical). Now, consider a user that is scheduled with $64 \times 64$ MIMO, $16$ carrier aggregation, and $256$-QAM modulation. This implies that the user observes $64 \times 16 \times 4 = 4096$ \textit{almost} independent channels with different statistical characteristics. 

A straightforward channel coding solution for such a scenario, as discussed above, is to encode the information independently across the given channels. For instance, suppose that a non-stationary sequence of $N$ channels is given and each of them can be used $l$ times in one block transmission. Then the straightforward solution is to encode information separately and independently across these channels, using codes of length $l$. However, from a finite-length analysis perspective, when $N$ and $l$ are comparable or when $N$ is much larger than $l$, such straightforward solutions can be extremely suboptimal in terms of gap to the average capacity. In other words, it is expected that a code of length $Nl$, carefully designed for the given sequence of channels, can significantly outperform the straightforward solution. In this paper, we study the extreme case where each channel is used only once in each transmission, i.e., $l=1$. However, the techniques introduced in this paper for non-stationary polar coding can be potentially generalized to cases with $l>1$ by combining our solution with existing polar coding schemes that have been proposed in stationary set-ups. 

\subsection{Related work}

The channel polarization problem for non-stationary channels was first considered in~\cite{Alsan}. In particular, it is shown in \cite{Alsan} that polarization happens by applying Ar{\i}kan's polarization transform. However, the proposed proof method is not powerful enough to conclude anything about the \textit{speed of polarization} ~\cite[Section IV]{Alsan}, although this notion is not explicitly defined in ~\cite{Alsan}. In particular, although a polar coding scheme is not introduced in \cite{Alsan}, one can not conclude the existence of an achievability scheme for non-stationary channels based on the proof of polarization provided in \cite{Alsan}. In Section\,\ref{sec:four}, we explicitly define the notion of \textit{speed of polarization} and the notion of \textit{fast} polarization, that can be applied in both the stationary and the non-stationary setups, from a finite-length analysis perspective. Note that an alternative notion of \textit{fast} polarization for stationary setups is introduced in the literature from an error-exponent analysis perspective \cite{tal2017simple}. 

In a stationary channel polarization process, independent and identical copies of the same channel $W$ are combined via a certain polarization transform \cite{Arikan}. For the resulting polar codes, it is shown that the required block length $N$ to guarantee a block error probability of $P_e$ is upper bounded by $c/\epsilon^\mu$,  where $\epsilon$ is the gap to the symmetric capacity, and $\mu$ and $c$ are constants \cite{hassani, guruswami}. Furthermore, it is shown in \cite{hassani} how to upper bound and lower bound this finite-length scaling behavior for different classes of channels. In other words, assuming that $N$ scales as $1/\epsilon^\mu$, upper bounds and lower bounds on the scaling exponent $\mu$ are derived in \cite{hassani}. These bounds were subsequently improved in \cite{GB,MUH}. In a sense, in this paper we extend such results to the case where the channels are independent but not identical by modifying Ar{\i}kan's channel polarization transform in many important respects.  

In another related line of work constructions of universal polar codes are studied \cite{csacsouglu2016universal,hassani2014universal}. Note that the construction of polar codes, as originally proposed by Ar{\i}kan \cite{Arikan}, is channel-dependent. In other words, the code has to be tailored to the channel in order to achieve the capacity. In universal polar coding setup, the goal is to construct polar-based codes that achieve the capacity for all the BMS channels with an equal capacity. This goal is attained, via different methods, in \cite{csacsouglu2016universal} and \cite{hassani2014universal}. The main difference between universal polar coding and non-stationary polar coding setups is the following. In universal polar coding, once the BMS channel is picked, from a given class of BMS channels, e.g., all BMS channels with equal capacity, the same channel is used, identically and independently, to carry all the coded bits. In other words, this setup still follows a stationary channel coding scenario. In the non-stationary polar coding setup, the channels to carry the coded bits can be all different and the code construction is still channel dependent. More specifically, the constructed polar codes, as described through this paper, depend on the given non-stationary sequence of channels. 

\subsection{Our contributions}

In this paper, we provide a precise notion of \textit{fast} polarization, from a finite-length analysis perspective, that can be used in both stationary and non-stationary set-ups.
Then we demonstrate a \emph{fast} polarization scheme for an arbitrary sequence of non-stationary BMS channels. To this end we modify Ar{\i}kan's channel polarization transform as the polarization process evolves. Certain permutations are applied to each sub-block of bit-channels at each polarization level before channel combining operations are applied. Also, the channel combining operations are skipped whenever it can not be guaranteed they improve a certain metric for polarization. The \emph{speed of polarization} is defined and bounds on the speed of the resulting polarization process are derived which show a \textit{fast} polarization. Furthermore, a two stage polarization process is shown which results in constructions of polar coding schemes that perform arbitrarily close to the average capacity of the channels. In particular, given $N = 2^n$ and $P_e > 0$, we construct a polar code of length $N$ with the gap to the average capacity $\epsilon$ and the probability of block error upper bounded by $P_e$ such that $N \leq \kappa / \epsilon^\mu$, where $\mu$ is a constant and $\kappa$ is a constant that depends on $\mu$ and $P_e$. Furthermore, a numerical upper bound on the parameter $\mu$ is established that is: $\mu \leq 7.34$ for non-stationary binary erasure channels and $\mu \leq 8.54$ for general non-stationary BMS channels. It is also shown that the encoding and decoding complexities of the constructed polar code preserve $O(N \log N)$ complexity of Ar{\i}kan's polar codes by adapting Ar{\i}kan's decoding trellis for successive cancellation decoding of polar codes. In an asymptotic sense, as the length of the proposed polar codes grows large, we achieve the limit of the average capacities, assuming that the limit exists. 

The rest of this paper is organized as follows. In Section\,\ref{sec:two} a brief overview of polar codes is provided. In Section\,\ref{sec:three} a \textit{modified} polarization scheme is shown by specifying certain permutations that we apply at each polarization level. In Section\,\ref{sec:four} the speed of polarization is defined and it is shown that the polarization scheme of Section\,\ref{sec:three} in combination with another modification results in a \textit{fast} polarization. In Section\,\ref{sec:five} a polar coding scheme is constructed via a two stage polarization and it is shown that it achieves the average capacity. Furthermore, bounds on the finite-length scaling exponent of the constructed codes are provided. Various numerical results on these bounds and simulation results are provided in Section\,\ref{sec:new}. Finally, the paper is concluded in Section\,\ref{sec:six}.

\section{Background}
\label{sec:two}

The \emph{channel polarization} phenomenon was discovered by Ar{\i}kan \cite{Arikan} and is based on a $2 \times 2$ polarization transform as the building block. The $2 \times 2$ polarization transform takes two independent copies of a binary-input discrete memoryless channel $W : \{0,1\} \to \sY$ and turns them into $(W^-,W^+)$, where $W^-$ and $W^+$ are also binary-input discrete memoryless channels. It is proved that by repeating the same process recursively, almost all the resulting bit-channels $W^{s^n}$, for $s^n \in \left\{+,-\right\}^n$, become either \emph{almost} noiseless with capacity very close to $1$ or \emph{almost} noisy, with capacity very close to $0$~\cite{Arikan}. 

The same polarization transform can be applied when the two underlying channels $W_1 : \{0,1\} \to \sY_1$ and $W_2 : \{0,1\} \to \sY_2$ are independent but not identical. The channel combining operations are defined as follows \cite{Arikan}: 
\begin{align}
\label{channel-com1}
&W_1 \boxcoasterisk W_2(y_1,y_2|u_1) = \frac{1}{2}  \sum_{u_2 \in \{0,1\}} W_1(y_1 | u_1 \oplus u_2) W_2(y_2|u_2), \\
\label{channel-com2}
&W_1 \circledast W_2(y_1,y_2,u_1|u_2) =  \frac{1}{2} W_1(y_1 | u_1 \oplus u_2) W_2(y_2|u_2),
\end{align}
where $y_1 \in \sY_1, y_2 \in \sY_2$, and $u_1,u_2 \in \left\{0,1\right\}$. In fact, when $W_1 = W_2 = W$, $W \boxcoasterisk W = W^-$ and $W \circledast W = W^+$. 

The problem of polarizing an arbitrary non-stationary sequence of channels $\left\{W_i\right\}_{i=1}^{\infty}$ was first considered in~\cite{Alsan}. Let $W_{n,i}$ denote the $i$-th bit-channel after $n$ levels of polarization for the sequence $\left\{W_i\right\}_{i=1}^N$, where $N = 2^n$. The parameters $N$ and $n$, where $N = 2^n$, are always used in this paper to specify the length of the polarization block and the constructed code. Let $ [\![N]\!]$ denote the set of positive integers less than or equal to $N$. It is proved in \cite[Theorem 2]{Alsan} that the fraction of non-polarized channels approaches zero, i.e., for every $0 < a < b <1$, 
\be{alsan}
\liminf_{N \rightarrow \infty} \frac{1}{N}\bigl| \left\{ i \in [\![N]\!]:  I(W_{n,i}) \in [a,b] \right\} \bigr| = 0.
\ee
However, this result does not conclude the existence of polar codes achieving the average symmetric capacity of $\left\{W_i\right\}_{i=1}^{\infty}$ \cite{Alsan}. 

Ar{\i}kan used the \emph{Bhattacharyya parameter} of $W$, denoted by $Z(W)$, to measure the reliability of $W$ and to upper bound the block error probability of polar codes. The Bhattacharyya parameter is defined as
$$
Z(W)
\,\ \deff\kern1pt
\sum_{y\in\sY} \!\sqrt{W(y|0)W(y|1)}.
$$
It is known that \cite{Arikan} \cite{Korada}
\be{eqZ1}
Z(W_1 \circledast W_2) = Z(W_1) Z(W_2),
\ee
and
\be{eqZ2}
Z(W_1 \boxcoasterisk W_2) \leq  Z(W_1) + Z(W_2) - Z(W_1)Z(W_2).
\ee
In particular, for a binary erasure channel (BEC) $W$ with erasure probability $p$, denoted by BEC$(p)$, we have $Z(W) = p$. Also, if $W_1 = \text{BEC}(p_1)$ and $W_2 = \text{BEC}(p_2)$, then both $W_1 \circledast W_2$ and $W_1 \boxcoasterisk W_2$ are also BECs with erasure probabilities $p_1p_2$ and $p_1+p_2-p_1p_2$, respectively \cite{Arikan}.

\section{A Modified Polarization Scheme}
\label{sec:three}
In this section we show a polarization process for non-stationary channels which combines Ar{\i}kan's polarization transform with a \textit{proper} sorting of the bit-channels at each polarization level. In the next section this polarization scheme is used, in combination with other modifications, to establish a \emph{fast} polarization.  

Suppose that a sequence of $N$ channels $\left\{W_i\right\}_{i=1}^{N}$, where $N = 2^n$, is given. A straightforward method of applying Ar{\i}kan's polarization transform results in the following process. In the first level of polarization, for $i \in [\![N/2]\!]$, $W_{2i-1}$ and $W_{2i}$ are combined and 
$$
W_{1,i} = W_{2i-1} \boxcoasterisk W_{2i}\ ,\ W_{1,N/2+i} = W_{2i-1} \circledast W_{2i}.
$$  
In the second level of polarization, the same procedure is applied to $\left\{W_{1,i}\right\}_{i=1}^{N/2}$ and $\left\{W_{1,i}\right\}_{i=N/2+1}^{N}$ in parallel. In general, in the $j$-th level of polarization, and for $i \in [\![N/2]\!]$, let 
$$
i = l2^{n-j} + r,
$$
where $r \in [\![2^{n-j}]\!]$ and $0 \leq l \leq 2^{j-1}-1$. Then we have
\begin{equation}
\label{ch-comb}
\begin{split}
W_{j,(2l)2^{n-j}+r} &= W_{j-1,2i-1} \boxcoasterisk W_{j-1,2i},\\
W_{j,(2l+1)2^{n-j}+r} &= W_{j-1,2i-1}  \circledast W_{j-1,2i},
\end{split}
\end{equation}
where the \emph{initial} channels at the level zero of polarization are $W_{0,i} = W_i$. 

\noindent
\textbf{Remark 1.\,}  Note that there is a subtle difference between how the indices of synthetic channels are arranged in \eq{ch-comb} versus how it is done in the non-stationary polarization process considered in \cite{Alsan}. In this paper, the polarization transform is considered together with the bit-reversal permutation, originally considered in \cite{Arikan}, whereas the bit-reversal permutation is not considered in \cite{Alsan}. Consider, for instance, the first level of polarization. With the bit-reversal permutation, the worse synthetic channels, the ones resulting from $. \boxcoasterisk .$ operations, are indexed by the elements of $[\![N/2]\!]$ and the better synthetic channels, the ones resulting from $. \circledast .$ operations, are indexed by the elements of $[\![N]\!] \setminus [\![N/2]\!]$. However, in \cite{Alsan}, which does not assume the bit-reversal permutation, the worse synthetic channels are indexed by odd numbers and the better synthetic channels are indexed by even numbers. Hence, in \cite{Alsan}, in the second level, synthetic channels with pairs of indices $(1,3)$, $(2,4)$, etc, are combined. Note that this is just another re-arrangement of indices at each level of polarization and does not change the logic behind the polarization process. 

In stationary channel polarization \cite{Arikan}, where all initial channels are identical, all the channels in each sub-block of length $2^{n-j}$, in the $j$-th level of polarization, are also identical. However, in our set-up, where the initial channels are not identical, one can combine channels in different orders. A deterministic process with certain permutations at each polarization level is defined next followed by specifying the permutations. 

\begin{definition}\label{polar1} (\emph{Deterministic Polarization 1}) Consider a polarization level $j$, $j=0,1,\dots,n-1$, with the sequence of channels $\left\{W_{j,i}\right\}$, $i \in [\![N]\!]$. This sequence is split into $2^j$ sub-blocks of equal length $2^{n-j}$ of consecutively indexed channels. Consider a sub-block $\left\{W_{j,l2^{n-j}+r}\right\}$, for $r\in [\![2^{n-j}]\!]$, where $0 \leq l\leq 2^{j}-1$. Let $\pi_{j,l}: [\![2^{n-j}]\!] \rightarrow [\![2^{n-j}]\!]$ denote a permutation which is determined according to a certain criteria. Then $W_{j,l2^{n-j}+r}$ is replaced by $W_{j,l2^{n-j}+\pi_{j,l}(r)}$ after which the channel combining operations of \eq{ch-comb} are applied to compute the synthetic channels of the next polarization level $j+1$. We refer to this polarization process, which is a combination of Ar{\i}kan's polarization process and certain permutations, as \textit{deterministic polarization} 1. For the sake of brevity, we often drop the permutations $\pi_{j,l}$ from the description of the \textit{deterministic polarization} 1 assuming that they are known and are applied, as described above, unless otherwise is mentioned. \end{definition}

In order to guarantee a \textit{fast} polarization, which will be demonstrated in the next section, the permutations $\pi_{j,l}$ are set in such a way that when they are applied to sub-blocks of channels, the sequence of Bhattacharyya parameters within each sub-block becomes non-increasing, i.e.,
$$
Z\bigl(W_{j,l2^{n-j}+\pi_{j,l}(r)}\bigr) \geq Z\bigl(W_{j,l2^{n-j}+\pi_{j,l}(r')}\bigr),
$$
if and only if $\pi_{j,l}(r) \leq \pi_{j,l}(r')$, where $r,r' \in [\![2^{n-j}]\!]$.

Next we describe how low complexity encoding and decoding of original polar codes \cite{Arikan} can be adapted to accommodate the \textit{deterministic polarization} 1 described above. Note that a memory of size $O(N \log N)$ must be added in order to save the permutations $\pi_{j,l}$ at both the encoder and the decoder. Since an efficient implementation of the low complexity successive cancellation (SC) decoder of polar codes requires only $O(N)$ space complexity, this means that the space complexity is increased to $O(N \log N)$ which is still \emph{quasi-linear} with $N$. 

The decoding process of polar codes consists of computations of likelihood ratios (LRs) through an $(n+1) \times N$ trellis reflecting the channel polarization operations \cite{Arikan}. The LRs at level zero are the channel reliabilities. At level $j$ of the decoding trellis, for $j \in [\![n]\!]$, the LR combination operations are done within $2^j$ sub-blocks of length $2^{n-j}$ each. The sub-blocks are processed successively but within each sub-block the LR combinations are done in parallel. The hard decisions on information bits are decided at the level $n+1$ of the trellis, where these hard decisions are propagated back through the trellis. The only modification that the \textit{deterministic polarization} 1 adds on top of this architecture is to apply the permutation $\pi_{j,l}$ to the $l$-th sub-block at level $j$ before processing that sub-block. Also, the inverse permutation $\pi_{j,l}^{-1}$ is applied to the hard decisions that are propagated back through the trellis. The resulting complexity of these changes will be negligible comparing to $O(N \log N)$ decoding complexity of polar codes. Also, since applying a permutation to a sequence can be done in one unit of time, the decoding latency $O(N)$ of polar codes is not affected by this modification. Similarly, the encoding process will be modified by invoking $\pi_{j,l}^{-1}$'s accordingly without affecting the total complexity or latency of the encoder.  

\noindent
\textbf{Remark 2.\,} In \cite{compound}, a scheme for permuting channels before applying the polarization transform, called \emph{compound polar codes}, is suggested. The scheme was shown to be applicable for BICM \cite{mahdavifar2016polar}. In a sense, \textit{deterministic polarization}~1 can be regarded as a \emph{recursive} compound channel polarization. In another related work in stationary settings, distinct channels in each polarization level, i.e., channels from different sub-blocks, are combined, as opposed to Ar{\i}kan's original method of combining identical channels, in order to attain universal polarization \cite{csacsouglu2016universal}. In \textit{deterministic polarization}~1, however, the channels to be combined are still within the same sub-block. Also, in \cite{twisted}, internal bit-channels induced by Ar{\i}kan's transformation are permuted (twisted) while being combined in order to obtain bit-channels with better performance.

\section{Fast Polarization for Non-Stationary Channels}
\label{sec:four}

Consider a non-stationary sequence of BMS channels $\left\{W_i\right\}_{i=1}^{N}$. The following metric is defined for this sequence of channels, which will be modified later in this section to define the \textit{speed of polarization}:
\be{Z-avg}
E(\left\{W_i\right\}_{i=1}^{N}) \,\deff\, \frac{1}{N} \sum_{i=1}^N z_i(1-z_i),
\ee
where $z_i = Z(W_i)$. Roughly speaking, if $E(.)$ is close to zero, then most of the channels have Bhattacharyya parameter either close to $0$ or close to $1$. Therefore, one would want to obtain a smaller $E(.)$ in order to have a better polarization. This metric is actually used in \cite{Arikan} in order to prove polarization, where a random process $\left\{Z_n\right\}$ is defined with respect to the polarization of a BMS channel $W$ \cite[Section IV]{Arikan} and then the convergence of the expected value $\cE\left\{Z_n(1-Z_n)\right\}$ is shown \cite[Proposition 9]{Arikan}. However, since the probability space of the random variable $Z_n$ is finite and the distribution is uniform, this expected value is equivalent to the average in \eq{Z-avg} computed over the $2^n$ synthetic channels after $n$ levels of polarization.
 
For simplicity, consider the case that initial $W_i$'s are identical and are BEC with erasure probability $p$, for some $p \in (0,1)$. For a BEC $W$, let $z = Z(W)$, $z^- = Z(W^-)$ and $z^+ = Z(W^+)$. Then by using the relations mentioned in Section\,\ref{sec:two} we have
\begin{align*}
\frac{z^+(1-z^+)+z^-(1-z^-)}{2z(1-z)} = 1 - z(1-z). 
\end{align*}
Note that $\sup_{z \in (0,1)} 1 - z(1-z) = 1$. This means that, intuitively, as the polarization process evolves, the \emph{speed of convergence} of $\cE\left\{Z_n(1-Z_n)\right\}$ reduces and one can not bound it away from $0$ while analyzing the finite-length performance of polar codes. In order to resolve this problem, functions of the type $z^b(1-z)^b$ are used in \cite{hassani} to replace quantities $z_i(1-z_i)$ in \eq{Z-avg}. In this section, we take a similar approach and define
\be{f-def}
f(z) \,\deff\,z^b(1-z)^b,
\ee
where the value of $b \in (0,1)$ is specified later. These types of functions are further extended in Section\,\ref{sec:new}, again similar to the approach in \cite{hassani}. Then the definition in \eq{Z-avg} is modified as follows:
\be{Z-avg2}
E(\left\{W_i\right\}_{i=1}^{N}) \,\deff\, \frac{1}{N} \sum_{i=1}^{N} f(z_i),
\ee
where $z_i = Z(W_i)$. 

Given the initial channels $W_{0,i} = W_i$, for $i\in [\![N]\!]$, $E_{j,n}$, for $j = 0,1,\dots,n$, is defined as follows:
\be{En-def}
E_{j,n} = E(\left\{W_{j,i} \right\}_{i=1}^{N}),
\ee
where $\left\{W_{j,i} \right\}_{i=1}^{N}$ is the sequence of synthetic channels in the $j$-th level of polarization according to the modified polarization scheme, discussed in Section\,\ref{sec:three}. 

\begin{definition}
\label{speed-def}
Given the initial channels $\left\{W_i\right\}_{i=1}^{N}$, the \emph{speed of polarization} at level $j$, where $j \in [\![n]\!]$, of the polarization process is defined as 
\be{speed-def1}
\eta_{j,n} \,\deff\, -\log E_{j,n}/E_{j-1,n},
\ee
and the \emph{average speed of polarization} at level $n$ is defined as 
\be{speed-def2}
\overline{\eta}_n \,\deff\, -\frac{1}{n}\log E_{n,n}. 
\ee
We call the polarization \emph{fast} if $\liminf_{n \rightarrow \infty} \overline{\eta}_n > 0$ for a given sequence of channels $\left\{W_i\right\}_{i=1}^{\infty}$.
\end{definition}

The goal is to find a lower bound for the speed of polarization $\eta_{j,n}$ and consequently for $\overline{\eta}_n$. This will require upper bounding expressions of the following form away from $1$:
\be{polar-rate}
\Delta_f(W,W') \,\deff\, \frac{f\bigl(Z(W \boxcoasterisk W')\bigr)+f\bigl(Z(W \circledast W')\bigr)}{f\bigl(Z(W)\bigr)+f\bigl(Z(W')\bigr)}.
\ee

Motivated by \eq{eqZ1}, \eq{eqZ2}, and \Lref{Z--lemma} on relations between Bhattacharyya parameters of combined channels, $g(z_1,z_2)$, for any $z_1,z_2 \in (0,1)$, is defined as follows:
\be{g-def}
g(z_1,z_2)\,\deff\,\sup_{z^- \in [\sqrt{z_1^2+z_2^2-z_1^2z_2^2},z_1+z_2-z_1z_2]} \frac{f(z_1z_2)+f(z^-)}{f(z_1)+f(z_2)}.
\ee
Then we have
\be{g-Delta}
\Delta_f(W,W') \leq g\bigl(Z(W),Z(W')\bigr).
\ee
Therefore, one can first upper bound $g\bigl(Z(W),Z(W')\bigr)$ and then use it in combination with \eq{g-Delta} in order to upper bound $\Delta_f(W,W')$. 

Note that since $f(.)$ is a convex function over $(0,1)$ with the maximum at $\shalf$, the supremum in \eq{g-def} happens at $z^- = \shalf$, if $\shalf \in [\sqrt{z_1^2+z_2^2-z_1^2z_2^2},z_1+z_2-z_1z_2]$ and otherwise, it happens at either of the extreme cases of $\sqrt{z_1^2+z_2^2-z_1^2z_2^2}$ or $z_1+z_2-z_1z_2$.

If we do not impose any restriction on $z_1$ and $z_2$, the supremum of $g(z_1,z_2)$ over all $z_1,z_2 \in (0,1)$ is at least $1$. This can be observed by setting $z_1 = p, z_2 = 1-p$, and letting $p \rightarrow 0$. Similarly, if $W$ is BEC$(p)$ and $W'$ is BEC$(1-p)$, then the supremum of $\Delta_f(W,W')$, defined in \eq{polar-rate}, is also at least $1$. In order to resolve this problem, the expression in \eq{polar-rate} is considered only when $Z(W_1)/Z(W_2)$ and/or $|Z(W_1) - Z(W_2)|$ are less than a certain threshold thereby making it possible to bound away the supremum of $g\bigl(Z(W_1), Z(W_2)\bigr)$ from $1$. This idea is elaborated through the rest of this section.

\begin{definition} \label{quant-def} (\emph{Fine Quantization of $(0,1)$})
The $[0,1]$ interval is quantized into roughly $c_1\log N + c_2$ sub-intervals as follows. The interval $[c,1-c]$ is quantized into equal sub-intervals of length $\lambda$, for some $c>0$. Then $[0,c]$ is quantized into $[0,c/2^m], [c/2^m,c/2^{m-1}],\dots, [c/4,c/2], [c/2,c]$, where $c/2^m \leq N^{-\tau}$ and $\tau$ is specified later. Therefore, we set
$$
m = \left \lceil \log c +\tau\log N \right \rceil.
$$
Similarly, $[1-c,1]$ is quantized into $[1-c,1-c/2], \dots,[1-c/2^{m-1},1-c/2^m],[1-c/2^m,1]$. 
\end{definition}

In total, we have
\be{intervals-num}
\begin{split}
2(m+1)+\frac{1-2c}{\lambda} & \leq 2\tau \log N + 2\log c + 4 + \frac{1-2c}{\lambda} \\
&= c_1\log N + c_2
\end{split}
\ee
sub-intervals in the quantization described above, where $c_1 = 2\tau$ and $c_2 = 2\log c + 4 + \frac{1-2c}{\lambda}$ are constants.  

Motivated by the quantization of $[0,1]$ described in Definition \ref{quant-def}, the function $h(z)$ for $z \in (0,1)$ is defined as follows:
\begin{equation}
\label{h-def}
h(z) \,\deff\,
\begin{cases}
\sup_{z' \in [z,2z]} &g(z,z')\ \text{if}\ z < c,
\\
\sup_{z' \in [z,\min(z+\lambda,1-c)]} &g(z,z')\ \text{if}\ c \leq z \leq 1-c,
\\
\sup_{z' \in [z,(1+z)/2]} &g(z,z')\ \text{if}\  z > 1-c,
\end{cases}
\end{equation} 
where $g(z_1,z_2)$ is defined in \eq{g-def}. The numerical analysis of the speed of polarization and its dependence on $\lambda$ and $c$ is discussed in Section\,\ref{sec:new}. 

It can be shown that $h(z)$ is continuous and bounded while its derivative is also bounded. Note that the right and the left derivatives may not be the same for $z = c$ and $z=1-c$, however, it does not affect the argument. Then,
\be{eta-def}
\eta\,\deff\, -\log \sup_{z \in (0,1)} h(z)
\ee
can be numerically estimated by quantizing $(0,1)$ with fine enough resolution. Also, note that $\eta$ depends on the parameter $b$ in the definition of $f(z) = z^b(1-z)^b$ and can be optimized by picking a proper $b$. In our analysis of the speed of polarization in this section and the scaling exponent of the constructed capacity-achieving codes in the next section we just need to assume that $\eta > 0$. Intuitively, larger $\eta$ leads to a larger speed of polarization. A more elaborate analysis of the numerical estimations of $\eta$ will be provided in Section\,\ref{sec:new}. 

The next lemma will be used in proving the main result of this section.

\begin{lemma}
\label{lemma1}
Let $W$ and $W'$ be two BMS channels with $z_1 = Z(W)$ and $z_2 = Z(W')$ and suppose that $z_1 \geq z_2$. If $z_1$ and $z_2$ satisfy either of the following three conditions: 
\begin{enumerate}
\item $c/2^m \leq z_2 < c$ and $z_1 \leq 2z_2$,
\item $c \leq z_2 < 1-c$ and $z_1 \leq \min(z_2 + \lambda, 1-c)$,
\item $1-c \leq z_2 \leq 1-c/2^m$ and $z_1 \leq (1+z_2)/2$,
\end{enumerate}
then
$$
\Delta_f(W,W') \leq 2^{-\eta},
$$
where $\Delta_f(.,.)$ is defined in \eq{polar-rate} and $\eta$ is given in \eq{eta-def}. 
\end{lemma}
\begin{proof}
By \eq{eqZ1}, \eq{eqZ2}, \eq{eqZ3}, and the definition of $g(.)$ in \eq{g-def} we have $\Delta_f(W,W') \leq g(z_1,z_2)$. By the definition of $h$ in \eq{h-def} and because $z_1$ and $z_2$ satisfy either of the three conditions we have
$$
g(z_1,z_2) \leq h(z_2) \leq \sup_{z \in (0,1)} h(z),
$$
which completes the proof.
\end{proof}

Before we state and prove the main theorem of this section, another modification on top of the \textit{deterministic polarization} 1, defined in Definition~\ref{polar1}, is discussed. We modify \textit{deterministic polarization} 1 by skipping certain channel combining operations. In particular, when we can not guarantee that $\Delta_f(W_{j,2i-1}, W_{j,2i}) \leq 1$, the combining operation of $W_{j,2i-1}$ and $W_{j,2i}$ is skipped. The indices for skipped operations are saved in a certain matrix $T$ specified next and will be used in both encoding and decoding of the constructed polar code, described in the next section, accordingly. 

\begin{definition}
\label{T-def}
The $n \times N/2$ matrix $T$, representing skipped channel combining operations, is defined as follows. For a pair of channels $W_{j-1,2i-1}$ and $W_{j-1,2i}$, $j \in [\![n]\!]$, $i \in [\![N/2]\!]$, if
\be{comb-cond1}
c/2^m \leq Z(W_{j-1,2i-1}), Z(W_{j-1,2i}) \leq 1-c/2^m,
\ee
where $c$ and $m$ are specified in the description of the fine quantization, in Definition\,\ref{quant-def}, and
\be{comb-cond}
\Delta_f(W_{j-1,2i-1},W_{j-1,2i}) \leq 1,
\ee
where $\Delta_f(.,.)$ is defined in \eq {polar-rate}, then $T(j,i) = 0$. Otherwise, $T(j,i) = 1$. 
\end{definition}

\begin{definition}\label{polar2} (\emph{Deterministic Polarization 2}) Consider the \textit{deterministic polarization} 1 defined in Definition~\ref{polar1} and a binary $n \times N/2$ matrix $T$. For each pair of channels $W_{j-1,2i-1}$ and $W_{j-1,2i}$ to be combined according to \eq{ch-comb}, we proceed with the channel combination operations only if $T(j,i) = 0$. Otherwise, 
$$
W_{j,(2l)2^{n-j}+r} = W_{j-1,2i-1}, W_{j,(2l+1)2^{n-j}+r}= W_{j-1,2i},
$$
where $i = l2^{n-j} + r$ and $r \in [\![2^{n-j}]\!]$. In other words, the channel combining operations are skipped for this particular pair of channels. We refer to this process as the \textit{deterministic polarization} 2.
\end{definition}
\textbf{Remark 3.\,} For the rest of this section, we consider the \textit{deterministic polarization} 2 with the matrix $T$, representing indices of skipped operations, as defined in Definition\,\ref{T-def}. Note that both $T$ and the polarization process are evolved jointly, i.e., when we are at the $j$-th level of the polarization process, the $j$-th row of $T$ is determined followed by performing $j$-th polarization level. Depending on the resulting bit-channels in the $j+1$-st level, the $j+1$-st row of $T$ is determined etc. However, the definitions of \textit{deterministic polarization} 2 and the matrix $T$ are separated since in the next section we will use \textit{deterministic polarization} 2 in combination with other matrices~$T$.  
\\\textbf{Remark 4.\,} Note that the skipped operations in the \textit{deterministic polarization} 2 only reduce the complexity of both encoding and decoding as the corresponding LR combining operations, related to the skipped channel combining operations, will be just transparent \cite{el2017relaxed}. The decoder needs to save the indices of the skipped operations which requires a memory of size at most $O(N \log N)$.

\begin{lemma}
\label{det-polar-lem}
In the \textit{deterministic polarization} 2, defined in Definition~\ref{polar2}, we have $E_{j,n} \leq E_{j-1,n}$, for $j \in [\![n]\!]$.
\end{lemma} 
\begin{proof}
According to the structure of the \textit{deterministic polarization} 2, when two channels $W_{j-1,2i-1}$ and $W_{j-1,2i}$ are combined, they satisfy $\Delta_f(W_{j-1,2i-1},W_{j-1,2i}) \leq 1$. If the combining operations are skipped, then the channels are the same at level $j$ and their Bhattacharyya parameters remain the same. Let $z_{j,i}$ denote $Z(W_{j,i})$. Then the lemma follows by summing $f(z_{j,i})$'s for all $i$'s and using the definition of $E_{j,n}$ in \eq{En-def}.    
\end{proof}

The following theorem summarizes the main result of this section. It shows that the \textit{modified} polarization scheme of Section\,\ref{sec:three}, with properly sorting bit-channels at each polarization level, together with carefully skipping certain channel combining operations results in a \textit{fast} polarization for non-stationary channels. Furthermore, a lower bound on the speed of polarization is provided that will be numerically computed in Section\,\ref{sec:new}.     

\begin{theorem}
\label{thm-main}
For $\rho < \frac{\eta}{\eta+1}$, where $\eta$ is defined in \eq{eta-def}, the parameters of the \textit{deterministic polarization} 2, defined in Definition \ref{polar2}, can be set such that for any $N=2^n$ and initial sequence of channels $\left\{W_{i}\right\}_{i=1}^N$, we have
$$
\overline{\eta}_n = -\frac{1}{n}\log E_{n,n} > \rho - \frac{c_{\rho}}{n},
$$
for a constant $c_{\rho} > 0$ depending only on $\rho$. 
\end{theorem}
\begin{proof}
Let $\rho< \frac{\eta}{\eta+1}$ and the sequence of channels $\left\{W_{i}\right\}_{i=1}^N$ be given. Let $\tau = \rho/b$, where the parameter $b$ is used in the definition of $f(.)$ in \eq{f-def}. Consider the \textit{deterministic polarization} 2 applied to the sequence of initial channels $\left\{W_{i}\right\}_{i=1}^N$. Let $z_{j,i} = Z(W_{j,i})$, for $i \in [\![N]\!]$ and $j=0,1,\dots,n$, where $W_{j,i}$'s are assumed to be sorted within each sub-block. Then define
\be{dj-def}
D_j = \frac{1}{N}\sum_{z_{j,i} < c/2^m \vee z_{j,i} > 1-c/2^m} f(z_{j,i}). 
\ee
Note that 
\be{dj-cond}
D_j \leq N^{-\tau b} = 2^{-\rho n}. 
\ee
Let $s = c_1\log N + c_2$ denote the number of sub-intervals in the fine quantization of $(0,1)$, described in Definition\,\ref{quant-def}. At level $j < n$, let $m_j$ denote the number of pairs of channels $(W_{j,2i-1},W_{j,2i})$ such that $c/2^m \leq z_{j,2i-1}, z_{j,2i} \leq 1 - c/2^m$ and, also, $z_{j,2i-1}$ and $z_{j,2i}$ do not satisfy any of the three conditions specified in \Lref{lemma1}. Then $z_{j,2i-1}$ and $z_{j,2i}$ must belong to different sub-intervals in the fine quantization, due to its structure described in Definition\, \ref{quant-def}. Since $z_{j,i}$'s are non-increasing within each sub-block, the number of pairs $(z_{j,2i-1}, z_{j,2i})$ that belong to different sub-intervals in the fine quantization is at most $s$, the total number of sub-intervals. Also, the number of sub-blocks is $2^j$. Hence, we have
\be{mj-eq}
m_j \leq s2^j.
\ee
Note that if $z_{j,2i-1}, z_{j,2i}$ satisfy one of the three conditions specified in \Lref{lemma1}, then $W_{j,2i-1}$ and $W_{j,2i}$ are combined in \textit{deterministic polarization} 2. This is by \Lref{lemma1} and knowing that $\eta > 0$. Then by applying \Lref{lemma1} to all such pairs of channels $(W_{j,2i-1},W_{j,2i})$ and noting that $D_j$ is non-decreasing with $j$ we get
\be{eq-thm1}
\begin{split}
E_{j+1,n} - D_{j+1} &\leq  E_{j+1,n} - D_j \leq \frac{m_j}{N} + (E_{j,n} - D_j)2^{-\eta}.
\end{split}
\ee
For $k \in [\![n]\!]$, we have
\begin{align}
\label{eq-thm5}
E_{k,n} & \leq D_k + \sum_{j=0}^{k-1} \frac{m_j 2^{-\eta(k-1-j)} }{N} + (E_{0,n}-D_0) 2^{-\eta k}\\
\label{eq-thm3}
&\leq D_k + \sum_{j=0}^{k-1} \frac{m_j }{N} + E_{0,n} 2^{-\eta k}\\
\label{eq-thm4}
& <D_k + s2^{-n+k}+E_{0,n}2^{-\eta k}\\
\label{eq-thm2}
&\leq 2^{-\rho n}+ s2^{-n+k}+2^{-\eta k},
\end{align}
where \eq{eq-thm5} is by by combining \eq{eq-thm1} for $j=0,1,\dots,k-1$, \eq{eq-thm3} is by noting that $D_0 \geq 0$, $\eta> 0$, and $k-1-j \geq 0$, \eq{eq-thm4} is by \eq{mj-eq} and summing a geometric series, and \eq{eq-thm2} is by \eq{dj-cond} and noting that $E_{0,n} \leq 1$.

Let $k = \left\lceil n/(\eta+1)\right\rceil$. We choose 
\be{crho-def}
c_{\rho} = \log \Bigl( 2+\max_{n\in \N} \bigl((c_1 n+c_2)2^{\rho n +\left\lceil n/(\eta+1)\right\rceil-n} \bigr) \Bigr),
\ee
which is well-defined because $\rho+1/(\eta+1)-1<0$. This is utilized to simplify \eq{eq-thm2} and to conclude that
$$
E_{k,n} < 2^{-n\rho + c_{\rho}}.
$$
By \Lref{det-polar-lem}, $E_{n,n} \leq E_{k,n}$ which completes the proof.

\end{proof}

\section{Code Construction and Finite-Length Analysis}
\label{sec:five}

In this section we show capacity achieving polar coding schemes for non-stationary channels and provide bounds on their finite-length scaling exponent.

In order to find explicit upper bounds on the Bhattacharryya parameters of the synthetic channels, we define an \emph{extremal} deterministic process. Roughly speaking, when two channels $W_1$ and $W_2$, with $z_1 = Z(W_1)$ and $z_2 = Z(W_2)$, are combined in the polarization process, we simply replace $Z(W_1 \boxcoasterisk W_2)$ by $z_1+z_2$, which, by \eq{eqZ2}, is an upper bound on $Z(W_1 \boxcoasterisk W_2)$. Due to having a simpler structure this process simplifies characterizing certain bounds in the finite-length regime. A more precise definition is provided next. 

\begin{definition}
\label{ext-proc}
Let $x_{0,i} \in [0,1]$, for $i \in [\![N]\!]$. A recursive process of combining $x_{j,i}$'s is defined as follows. For $j \in [\![n]\!]$, the sequence $\left\{x_{j-1,i}\right\}_{i=1}^N$ is split into $2^{j-1}$ equal sub-blocks, each containing $N/2^{j-1}$ consecutive elements. Then $x_{j-1,i}$'s are permuted in such a way that they become sorted in a non-increasing order within each sub-block, before being combined for the $j$-th level as described next. For $i \in [\![N/2]\!]$, let 
$$
u = x_{j-1,2i-1},\ \text{and}\ v = x_{j-1,2i}.
$$
Also, let $i = l2^{n-j} + r$, where $r \in [\![2^{n-j}]\!]$ and $0 \leq l \leq 2^{j-1}-1$. If $u,v \leq 1$ or $u,v \geq 1$, we let 
\be{z-comb1}
x_{j,(2l)2^{n-j}+r} = u+v,\ \text{and}\ x_{j,(2l+1)2^{n-j}+r} = uv,
\ee
and otherwise, i.e., when $u >1> v$, we let
\be{z-comb2}
x_{j,(2l)2^{n-j}+r} = u,\ \text{and}\ x_{j,(2l+1)2^{n-j}+r} = v.
\ee
In other words, the combining operation is skipped in this case. This finite deterministic process is referred to as the extremal deterministic process associated with $\left\{x_{0,i}\right\}_{i=1}^N$. 
\end{definition}

The function $q:\R^+\cup\left\{0\right\} \rightarrow \R^+\cup\left\{0\right\} $ is defined as follows:
\begin{equation}
\label{q-def}
q(x) \,\deff\,
\begin{cases}
x(2-x)\ &\text{if}\ x \leq 1,
\\
1 \ &\text{if}\ x > 1.
\end{cases}
\end{equation} 

In the next lemma we show that $q(.)$ can be used as a \emph{potential function} in an extremal deterministic process. 
\begin{lemma}
\label{pot-lemma}
In an extremal deterministic process, defined in Definition~\ref{ext-proc}, we have
$$
\sum_{i=1}^N q(x_{j,i}) \leq \sum_{i=1}^N q(x_{j-1,i}).
$$
for $j\in [\![n]\!]$. 
\end{lemma}
\begin{proof}
For $u,v \leq 1$, if $u+v \leq 1$, we have
\be{pot1}
q(u) + q(v) - q(uv) - q(u+v) = u^2 v^2 \geq 0,
\ee
and if $u+v > 1$, $q(u+v) =1$ and
\be{pot2}
q(u) + q(v) - q(uv) - 1 = (1-u)(1-v)(u+v+uv-1) \geq 0.
\ee
If $u,v \geq 1$, then
\be{pot3}
q(u) + q(v) - q(uv) - q(u+v) = 0.
\ee
Note that if $u>1>v$, a similar inequality does not always hold but according to \eq{z-comb2}, the combining operations are skipped in this case. Therefore, the lemma follows by using \eq{pot1}, \eq{pot2}, and \eq{pot3} together with definitions of combining operations in \eq{z-comb1} and \eq{z-comb2} in an extremal deterministic process.
\end{proof}
\begin{corollary}
\label{pot-cor2}
In an extremal deterministic process, defined in Definition~\ref{ext-proc}, we have
\be{det-lemma2-1}
\Bigl|\left\{i: i \in [\![N]\!], x_{j,i} \geq 1\right\}\Bigr| \leq \sum_{i=1}^N q(x_{0,i}),
\ee
for $j \in [\![n]\!]$. 
\end{corollary}
\begin{proof}
Since $q(x) = 1$ for $x\geq 1$ and by \Lref{pot-lemma} we have
$$
\Bigl|\left\{i: i \in [\![N]\!], x_{j,i} \geq 1\right\}\Bigr| \leq \sum_{i=1}^N q(x_{j,i}) \leq \sum_{i=1}^N q(x_{0,i}).
$$
\end{proof}
For a non-negative integer $i$, let $s(i)$ denote the number of $1$'s in the binary representation of $i-1$. In other words, let $b_1b_2\dots b_n$ denote the binary representation of $i-1$. Then
\be{s-def}
s(i)\,\deff\, \sum_{i=1}^n b_i.
\ee
Also, assuming $n$ is fixed, $s_j(i)$ denote the number of $1$'s in the leftmost $j$ bits in the $n$-bit binary representation of $i-1$, i.e.,
\be{sj-def}
s_j(i)\,\deff\, \sum_{i=1}^j b_i.
\ee
\textbf{Remark 5.\,} It can be observed that in the evolution of a bit-channel $W_{j,i}$, for $i \in [\![N]\!]$, the number of polarization levels where the operation $\circledast$ is applied is $s_j(i)$ and in the rest of $j-s_j(i)$ levels the operation $\boxcoasterisk$ is applied. Therefore, it is expected that the quality of the bit-channels $W_{n,i}$'s, e.g., their Bhattacharyya parameters, can be well bounded in terms of $s(i)$. Similarly, in an extremal deterministic process, $x_{n,i}$'s can be bounded in terms of $s(i)$ and the initial parameters. The following lemma provides such a bound, which is one of the key elements in establishing upper bounds on the probability of error of constructed polar coding schemes. 
\begin{lemma}
\label{key-lemma}
Let $x_i = x_{0,i} \in (0,\shalf)$, for  $i \in [\![N]\!]$. Consider an extremal deterministic process, defined in Definition~\ref{ext-proc}, associated with $x_i = x_{0,i}$. Then the criteria for permutations and skipped operations in the process can be set such that for $j \in [\![n]\!]$,
\be{det-lemma2-1}
\Bigl|\left\{i: i \in [\![N]\!], x_{j,i} \leq 2^{-2^{s_j(i)}+a_{j,i}} \right\}\Bigr| \geq N - 4 \sum_{i=1}^N x_i(1-x_i),
\ee
where $s_j(i)$ is defined in \eq{sj-def} and $\sum_{i=1}^N a_{j,i} \leq 3^j$.
\end{lemma}
\begin{proof}
Let $y_i = y_{0,i} = 2x_i <1$, for  $i \in [\![N]\!]$. Consider the extremal deterministic processes $\left\{x_{j,i}\right\}$ and $\left\{y_{j,i}\right\}$. Suppose that there were no skipped operations, i.e., the processes evolve only according to \eq{z-comb1}. Then it can be observed that at the $j$-th level of the process we would have
\be{key1}
y_{j,i} = 2^{2^{s_j(i)}} x_{j,i},
\ee
where $s_j(i)$ is defined in \eq{sj-def}. 

Next, the effect of skipped operations on \eq{key1} is characterized. We modify the criteria for skipped operations in $\left\{x_{j,i}\right\}$ such that it follows the pattern of skipped operations of the process $\left\{y_{j,i}\right\}$, i.e., the operations in \eq{z-comb1} are applied if both $y_{j-1,2i-1}$ and $y_{j-1,2i}$ are less than $1$ or greater than $1$ and otherwise, \eq{z-comb2} is applied. Also, the same permutations that are applied to the sub-blocks of $\left\{y_{j,i}\right\}$ to make them non-increasing within each sub-block of length $2^{n-j}$ are applied to the process $\left\{x_{j,i}\right\}$. 

The parameters $a_{j,i}$, for $j = 0,1,\dots,n$ and $i \in [\![N]\!]$, are recursively derived such that 
\be{key2}
y_{j,i} \geq 2^{2^{s_j(i)}-a_{j,i}} x_{j,i}.
\ee 
We confirm that \eq{key2} holds by induction on $j$. For the base of induction $j=0$, we define $a_{0,i} = 0$ and \eq{key2} holds. In level $j$, consider $u = y_{j,2i-1}$ and $v = y_{j,2i}$ to be combined for the next level. Let $i = l2^{n-j-1} + r$, where $r \in [\![2^{n-j-1}]\!]$ and $0 \leq l \leq 2^j-1$, and 
$$ 
i_1 = (2l)2^{n-j-1}+r,\ i_2 = (2l+1)2^{n-j-1}+r
$$ 
denote the indices of the combined $u$ and $v$ in the next level, i.e., if $u,v \geq 1$ or $u,v \leq 1$, then
\be{key3}
y_{j+1,i_1} = u+v,\ \text{and}\ y_{j+1,i_2} = uv.
\ee
In this case, we set
\be{key4}
a_{j+1,i_1} = \max(a_{j,2i-1},a_{j,2i}),\ \text{and}\  a_{j+1,i_2} = a_{j,2i-1}+a_{j,2i}.
\ee
By induction hypothesis \eq{key2} holds for $(j,2i-1)$ and $(j,2i)$. Also, note that 
$$
s_{j+1}(i_1) = s_j(2i-1) = s_j(2i) = s_{j+1}(i_2)-1.
$$
This together with the induction hypothesis, \eq{key3}, and \eq{key4} imply that \eq{key2} holds for $(j+1,i_1)$ and $(j+1,i_2)$. 

If $u > 1 > v$, then
$$
y_{j+1,i_1} = u,\ \text{and}\ y_{j+1,i_2} = v.
$$
In this case we set
\be{key5}
a_{j+1,i_1} = a_{j,2i-1},\ \text{and}\  a_{j+1,i_2} = a_{j,2i} + 2^{s_j(2i)}.
\ee
Then since $s_{j+1}(i_1) = s_j(2i-1)$, induction hypothesis holds for $(j+1,i_1)$. Also, we have
$$
2^{s_{j+1}(i_2)} - a_{j+1,i_2} = 2^{s_{j}(2i)+1} - a_{j,2i} - 2^{s_j(2i)} = 2^{s_{j}(2i)} - a_{j,2i},
$$
and therefore, induction hypothesis holds for $(j+1,i_2)$, too. 

What remains is to upper bound $\sum_{i=1}^N a_{j,i}$. Let
$$
A_j\,\deff\, \sum_{i=1}^N a_{j,i}.
$$
Note that at level $j$, there are $2^j$ sub-blocks of length $2^{n-j}$ each. Also, $s_j(i)$ is the same for all indices $i$ in one sub-block. Let $k$ denote the number of skipped operations at level $j$. Since $y_{j,i}$'s are sorted within each sub-block, there is at most one skipped operation within each sub-block and hence, we have $k \leq 2^j$. Let $i_1,\dots,i_k$ denote the even indices of the pairs in skipped operations. Then we have
\be{key6}
\sum_{l=1}^k 2^{s_j(i_l)} \leq \sum_{i=0}^{j} {j \choose i} 2^i = 3^j.
\ee
By \eq{key4}, \eq{key5}, and \eq{key6} we have
$$
A_{j+1} \leq 2A_j + 3^j.
$$
And by induction on $j$, $A_j \leq 3^j$. This together with \eq{key2}, \Cref{pot-cor2} applied to the process $\left\{y_{j,i}\right\}$, and noting that $q(y_i) = 4x_i(1-x_i)$ lead to the proof of lemma. 
\end{proof}

Next, we state the main result of this section. In the proof of the next theorem, a polar code construction along with a bound on its finite-length scaling is shown. The proof sketch and the idea behind the construction is described as follows. The construction is based on a two-stage \textit{deterministic polarization}. In the first stage, \textit{deterministic polarization} 2, defined in Definition\,\ref{polar2}, is invoked to obtain \emph{almost good} synthetic channels, whose fraction is close to the average capacity, promised by \Tref{thm-main}. In other words, \Tref{thm-main} is used at the last level of polarization in the first stage. In the second stage of the polarization, only the almost good synthetic channels of the first stage are further combined and the rest of the synthetic channels are discarded, i.e., they will not carry information bits. The channel combining operations in the second stage follow a certain extremal deterministic process, defined in Definition\,\ref{ext-proc}. At the end of the second polarization stage, a negligible fraction of synthetic channels are removed and it is shown that the remaining ones have Bhattacharyya parameters bounded by $P_e/N$, where $P_e$ is the target block error probability and $N$ is the block length. The end result is a polar code with a rate close the average capacity and with a block error probability bounded by $P_e$. These steps are clarified in the proof of the next theorem. 

\begin{theorem}
\label{thm-main2}
For a sequence of BMS channels $\left\{W_i\right\}_{i=1}^{N}$, target block error probability $P_e \in (0,1)$, and a constant $\mu$ such that
\be{mu-def}
\mu > 2+\log 3 + \frac{1}{\eta},
\ee
where $\eta$ is defined in \eq{eta-def}, we construct a polar code of length $N = 2^n$ and rate $R$ guaranteeing a block error probability of at most $P_e$ for transmission over $\left\{W_i\right\}_{i=1}^{N}$ such that 
\be{eq-main21}
N \leq \frac{\kappa}{(\IN - R)^{\mu}},
\ee
where $\kappa$ is a constant depending on $P_e$ and $\mu$, and $\IN$ is the average of the symmetric capacities $I(W_i)$, for $i \in [\![N]\!]$.
\end{theorem}
\begin{proof}
The parameter $\rho \in \R^+$ is picked such that 
\be{rho-def1}
\mu > 1+\log 3 + \frac{1}{\rho} > 2+\log 3 + \frac{1}{\eta}.
\ee
In particular, we pick
\be{rho-def2}
\rho\,\deff\,\frac{2}{\mu+\frac{1}{\eta}-\log 3}.
\ee
Therefore, $\rho < \eta/(\eta+1)$. Let 
\be{n1-def}
n_1 = \left\lceil \frac{1}{1+\rho(1+\log 3)} n\right\rceil,
\ee
and
\be{n2-def}
n_2 = n - n_1 > \frac{\rho(1+\log 3)}{1+\rho(1+\log 3)} n - 1.
\ee
Let also $N_1 = 2^{n_1}$ and $N_2 = 2^{n_2}$. We exploit \Lref{partition-lemma} in Appendix to partition the set of channels $\left\{W_i: i \in [\![N]\!]\right\}$ into $N_2$ subsets $A_1,A_2,\dots,A_{N_2}$ of size $N_1$ such that
\be{main2-1}
\overline{I}_{N_1,k} \geq \IN - 2^{-n_1},
\ee
where $\overline{I}_{N_1,k}$ is the average of symmetric capacities of the channels in the set $A_k$, for $k \in [\![N_2]\!]$. 

In the first stage of polarization the \textit{deterministic polarization} 2, defined in Definition\,\ref{polar2}, is applied to each of the sets $A_k$ separately. The permutations $\pi^{(k)}_{j,l}$, for $j\in [\![n_1]\!]$, $k \in [\![N_2]\!]$ and $0 \leq l\leq 2^{j}-1$ are set in such a way that the sequence of Bhattacharyya parameters of bit-channels within the $l$-th sub-block of length $2^{n_1-j}$ in the $j$-th level of polarization of the set $A_k$ becomes non-increasing. Also, the sets $T_k$ representing the indices of the skipped operations in the polarization of $A_k$ are set according to Definition\,\ref{T-def}. Let $W^{(k)}_{j,i}$ denote the $i$-th bit-channel in the $j$-th polarization level of the set $A_k$ and $z^{(k)}_{j,i}$ denote its Bhattacharyya parameter. 

By \Tref{thm-main} there exists a constant $c_\rho$ such that
\be{main2-2}
\frac{1}{N_1} \sum_{i=1}^{N_1} f(z^{(k)}_{n_1,i}) \leq 2^{-n_1\rho + c_\rho}.
\ee
Then by applying \Cref{C-key} to $\left\{W^{(k)}_{j,i}\right\}_{i=1}^{N_1}$ with the choice of $t<\shalf$ we have 
\begin{align}
\label{main2-3-1}
\frac{1}{N_1} \Bigl|\left\{i \in [\![N_1]\!] : z^{(k)}_{n_1,i} < \frac{1}{2} \right\}\Bigr| &\geq \overline{I}_{N_1,k} - \frac{2}{tN_1}\sum_{i=1}^{N_1} f(z^{(k)}_{n_1,i})\\
\label{main2-3-2}
&\geq \IN - 2^{-n_1} - \frac{1}{t}2^{-n_1\rho + c_\rho+1}\\
\label{main2-3-3}
&\geq \IN - d_1 2^{-n_1\rho}\\
\label{main2-3}
&\geq \IN - d_1 N^{-1/\mu},
\end{align}
where \eq{main2-3-1} is by \Cref{C-key}, \eq{main2-3-2} is by \eq{main2-1} and \eq{main2-2}, \eq{main2-3-3} is by $d_1$ being defined as
\be{d1-def}
d_1 \,\deff\, \frac{1}{t}2^{c_{\rho}+1}+1,
\ee
and \eq{main2-3} is by \eq{rho-def1} and \eq{n1-def}.  

Let
\be{m-main-def}
M = \left\lfloor N_1 (\IN - d_1 N^{-1/\mu})\right\rfloor.
\ee
Without loss of generality, we can assume that
$$
z^{(k)}_{n_1,i} < \frac{1}{2},
$$
for $k \in [\![N_2]\!]$ and $i \in [\![M]\!]$. In the second stage of polarization, we will apply a polarization process to $\left\{W^{(k)}_{n_1,i}\right\}_{k=1}^{N_2}$, for $i \in [\![M]\!]$. Note that the actual indices of these good bit-channels within each of the polarization blocks, corresponding to sets $A_k$, does not matter as long as they are in increasing order. In other words, roughly speaking, the $i$-th good bit-channels in the polarization of the sets $A_k$ will be combined in the second stage of polarization. All the other $N_2(N_1 - M)$ bit-channels will not be further polarized and will be frozen to zeros. 

Fix $i_0 \in [\![M]\!]$ and for ease of notation let $W_k = W^{(k)}_{n_1,i_0}$. Let $x_{0,k} = Z(W_k)$ and $y_{0,k} = 2x_{0,k}$. Consider an extremal deterministic process $\left\{y_{j,k}\right\}$, for $j=0,1,\dots,n_2$ and $k\in [\![N_2]\!]$, defined in Definition\,\ref{ext-proc} with the initial values $\left\{y_{0,k}\right\}$. Also, consider another extremal deterministic process $\left\{x_{j,k}\right\}$ with the initial values $\left\{x_{0,k}\right\}$, where permutations and indices of skipped operations follow the process $\left\{y_{j,k}\right\}$. By \Lref{key-lemma}, for any $j \in [\![n_2]\!]$,
\be{main2-5}
\begin{split}
&\Bigl|\left\{k: k \in [\![N_2]\!], x_{j,k} \leq 2^{-2^{s_j(k)}+a_{j,k}} \right\}\Bigr|\\
& \geq N_2 - 4 \sum_{k=1}^N Z(W_k)\bigl(1-Z(W_k)\bigr),
\end{split}
\ee
where $s_j(k)$ is defined in \eq{sj-def} and $\sum_{k=1}^N a_{j,k} \leq 3^j$. Let 
\be{l-main-def}
l = \left\lfloor \frac{n_2}{1+\log 3}\right\rfloor,
\ee
and
\be{alpha-def}
\alpha = 1+ \frac{1}{\rho(1+\log 3)},\ \beta = -\log P_e + \alpha.
\ee
By \Lref{2power-lemma} we have 
\be{main2-6}
\begin{split}
&\log \Bigl|\left\{k \in [\![N_2]\!]: 2^{s_l(k)} - a_{l,k} \leq \alpha n_2 + \beta\right\}\Bigr| \\
&\leq n_2-l+p(n_2,\alpha,\beta),
\end{split}
\ee
where $p(n,\alpha,\beta)$ is defined in \eq{p-def}. Let $\gamma = 1+\log 3$ and define 
\be{qn-def}
\begin{split}
q(n)\,\deff\, n(\frac{1}{\mu} - \frac{\rho}{1+\gamma \rho}) + 1+ \frac{1}{\gamma} + p(\frac{\rho\gamma}{1+\rho\gamma}n,\alpha,\beta),
\end{split}
\ee
where $\alpha$ and $\beta$ are defined in \eq{alpha-def}. Note that by \eq{rho-def1} the coefficient of $n$ in the definition of $q(n)$ is negative. Also, $p$ is bounded by a polynomial of $\log n$. Therefore, 
$$
\lim_{n \rightarrow \infty} q(n) = -\infty
$$ 
and 
\be{d2-def}
d_2 = \sup_{n \in \N} q(n)
\ee
is a well-defined constant. Therefore, by the choice of $l$ in \eq{l-main-def} we have
\be{main2-6-2}
l - p(n_2,\alpha,\beta) \geq \frac{n}{\mu} - d_2
\ee
Also, by the choice of $\alpha$ and $\beta$ in \eq{alpha-def} we have
\be{main2-7}
2^{-\alpha n_2 - \beta} \leq \frac{P_e}{N}.
\ee
Then we have
\begin{align}
\notag
&\Bigl|\left\{k: k \in [\![N_2]\!], x_{l,k} \leq \frac{P_e}{N} \right\}\Bigr|\\
\label{main2-8-1}
&\geq \Bigl|\left\{k: k \in [\![N_2]\!], x_{l,k} \leq 2^{-\alpha n_2 - \beta} \right\}\Bigr|\\
\label{main2-8-2}
&\geq \Bigl|\left\{k: k \in [\![N_2]\!], x_{l,k} \leq 2^{-2^{s_l(k)}+a_{l,k}}\right\}\Bigr|\\
\notag
&- \Bigl|\left\{k \in [\![N_2]\!]: 2^{s_l(k)} - a_{l,k} \leq \alpha n_2 + \beta\right\}\Bigr|\\
\label{main2-8-3}
& \geq N_2 - 4 \sum_{k=1}^N Z(W_k)\bigl(1-Z(W_k)\bigr) - 2^{n_2-l+p(n_2,\alpha,\beta)}\\
\label{main2-8}
& \geq N_2 - 4 \sum_{k=1}^N Z(W_k)\bigl(1-Z(W_k)\bigr) - N_2 2^{-\frac{n}{\mu}+d_2},
\end{align}
where \eq{main2-8-1} is by \eq{main2-7}, \eq{main2-8-2} is by the union bound, \eq{main2-8-3} is by \eq{main2-5} and \eq{main2-6}, and \eq{main2-8} is by \eq{main2-6-2}. 

Recall that \eq{main2-8} holds for any $i_0 \in [\![M]\!]$, where we dropped the index $i_0$ from the derivations for notation convenience. Let $x_{j,k}$ be re-indexed by $x_{j,k+N_2(i_0-1)}$, for $i_0 \in [\![M]\!]$. By summing \eq{main2-8} for all $i_0$ we have
\be{main2-9}
\begin{split}
&\Bigl|\left\{i: i \in [\![MN_2]\!], x_{l,i} \leq \frac{P_e}{N} \right\}\Bigr|\\
&\geq MN_2 - 4 \sum_{i_0 \in [\![M]\!], k \in [\![N_2]\!]} z^{(k)}_{n_1,i_0}(1-z^{(k)}_{n_1,i_0}) - MN_2 2^{-\frac{n}{\mu}+d_2}.
\end{split}
\ee 
Also, by summing \eq{main2-2} over all $k \in [\![N_2]\!]$ and noting that since $b\in (0,1)$, $f(z)= z^b(1-z)^b \geq z(1-z)$, for any $z \in [0,1]$, we have
\be{main2-10}
\begin{split}
\sum_{i_0 \in [\![M]\!], k \in [\![N_2]\!]} z^{(k)}_{n_1,i_0}(1-z^{(k)}_{n_1,i_0}) & \leq  \sum_{i_0 \in [\![M]\!], k \in [\![N_2]\!]} f(z^{(k)}_{n_1,i_0})\\ 
& \leq N2^{-n_1\rho + c_\rho} \\
& \leq 2^{c_\rho} N^{1-\frac{1}{\mu}}.
\end{split}
\ee
By combining \eq{main2-10} and \eq{main2-9} we have
\be{main2-11}
\Bigl|\left\{i: i \in [\![MN_2]\!], x_{l,i} \leq \frac{P_e}{N} \right\}\Bigr| \geq MN_2 - (2^{c_\rho+2} +2^{d_2})N^{1-\frac{1}{\mu}},
\ee
where we also used $M \leq N_1$. By plugging $M$ from \eq{m-main-def} in \eq{main2-11} and normalizing both sides by $N$ we get
\be{main2-12}
\frac{1}{N}\Bigl|\left\{i: i \in [\![MN_2]\!], x_{l,i} \leq \frac{P_e}{N} \right\}\Bigr| \geq \IN - d_3 N^{-\frac{1}{\mu}},
\ee
where $d_3$ is defined as
\be{d3-def}
d_3\,\deff\,d_1 + 2^{c_\rho+2} +2^{d_2} +1,
\ee
where $d_1$ is defined in \eq{d1-def}, $d_2$ is defined in \eq{d2-def}, $\rho$ is defined in \eq{rho-def2}, and $c_\rho$, as a function of $\rho$, is defined in \eq{crho-def}. Note that $d_3$ is a constant that depends only on $\mu$ and $P_e$.

In the second stage of polarization, \textit{deterministic polarization} 2, defined in Definition,\,\ref{polar2}, is applied to $\left\{W^{(k)}_{n_1,i_0}\right\}_{k=1}^{N_2}$, for $i_0 \in [\![M]\!]$, where the channels are combined according to the index $k$. Let the bit-channels at the $j$-th level of the second stage be denoted by $\left\{W^{(k)}_{n_1+j,i_0}\right\}_{k=1}^{N_2}$. The permutations and the indices for skipped operations follow the extremal process $\left\{y_{j,k}\right\}$ described earlier in the proof for the fixed $i_0$. The polarization stops at level $l$, where $l$ is given in \eq{l-main-def}. One can alternatively assume that the process continues but all the operations for $l <j \leq n_2$ are skipped. To simplify the notation let the bit-channels $W^{(k)}_{n_1+j,i_0}$ be re-indexed by $W_{n_1+j,i_0+N_2(k-1)}$. The polar code is constructed by picking all the bit-channels at the final level $n = n_1+n_2$ with the Bhattacharyya parameter bounded by $P_e/N$. The rate $R$ of the code is given by
\be{R-def}
R = \frac{1}{N}\left\{i: i \in [\![mN_2]\!],\ Z(W_{n,i}) \leq \frac{P_e}{N}\right\}.
\ee
By \Lref{det-lemma1}
$$
Z(W_{j,i}) \leq x_{j,i},
$$
for $j \in [\![n_2]\!]$ and $i \in [\![mN_2]\!]$. Using this together with \eq{R-def} and \eq{main2-12} we have
\be{R-bound}
\begin{split}
R &\geq \frac{1}{N}\Bigl|\left\{i: i \in [\![MN_2]\!], x_{l,i} \leq \frac{P_e}{N} \right\}\Bigr|\\
& \geq \IN - d_3 N^{-\frac{1}{\mu}}.
\end{split}
\ee
Then \eq{R-bound} can be re-written as
$$
N \leq \frac{\kappa}{(\IN - R)^{\mu}},
$$
where
\be{kappa-def}
\kappa\,\deff\,d_3^{\mu}
\ee
and $d_3$ is defined in \eq{d3-def}. Note that $\kappa$ depends only on $\mu$ and $P_e$. Also, the probability of block error, under successive cancellation decoding, is bounded by the sum of Bhattacharyya parameters of the selected bit-channels which is bounded by $P_e$, similar to what is shown for original polar codes \cite{Arikan}. 
\end{proof}
\textbf{Remark 6.\,} The low complexity $O(N \log N)$ architecture of successive cancellation decoding of original polar codes can be adapted for the polar code constructed via a two stage polarization explained in the proof of \Tref{thm-main2}. The decoding trellis can be also separated into two stages to reflect the two stages of polarization. The first of stage of the decoding trellis consists of $N_1$ trellises, each of size $(n_2+1) \times N_2$. These trellises are run in parallel while the permutations $\pi^{(k)}_{j,l}$ are applied accordingly as described in Section\,\ref{sec:three}. Also, the indices for the skipped operations, saved in $T_k$'s, are taken into account and the corresponding LR combination operations are skipped. The second stage of decoding trellis consists of $M$ trellises, where $M$ is given in \eq{m-main-def}, each of size $(n_1+1) \times N_1$. These trellises are run successively while the permutations and skipped operations, derived from the extremal deterministic process as described in the proof of \Tref{thm-main2}, are also applied. When the decoding process in the first stage trellises arrive at the $i$-th selected index in the level $n_1+1$, for $i \in [\![M]\!]$, the $i$-th trellis of the second stage is run. When it is done, the hard decisions are passed to the first stage trellises and they continue to run until they arrive at the $i+1$-th selected index and so on. The complexity of the total decoding process is $O(N \log N)$. 

\begin{theorem}
\label{thm-main3}
Given a non-stationary sequence of independent BMS channels $\left\{W_i\right\}_{i=1}^{\infty}$ and assuming the average capacity
$$
\overline{I}(\left\{W_i\right\}_{i=1}^{\infty}) = \lim_{N\rightarrow \infty} \frac{1}{N}\sum_{i=1}^N I(W_i)
$$
is well-defined, there exists a sequence of polar codes $\left\{\cC_i\right\}_{i=1}^{\infty}$, with length $N_i$, rate $R_i$, and probability of block error $P_i$ under a low complexity $O(N_i \log N_i)$ SC decoder, such that
$$
\lim_{i \rightarrow \infty} R_i = \overline{I}(\left\{W_i\right\}_{i=1}^{\infty}),\ \text{and}\ \lim_{i \rightarrow \infty} P_i = 0,
$$
when the $j$-th bit of $\cC_i$ is transmitted over $W_j$, for $j \in [\![N_i]\!]$.
\end{theorem}
\begin{proof}
Consider two arbitrary sequences $\left\{\epsilon_i\right\}_{i=1}^{\infty}$ and $\left\{\overline{P}_i\right\}_{i=1}^{\infty}$ such that
$$
\lim_{i \rightarrow \infty} \epsilon_i = 0,\ \text{and}\ \lim_{i \rightarrow \infty} \overline{P}_i = 0.
$$
Fix $\mu$ such that 
$$
\mu > 2+\log 3 + \frac{1}{\eta},
$$  
where $\eta$ is given in \eq{eta-def}. For every $i \in \N$, let
\be{ni-def}
n_i = \left \lceil  \log \kappa (\overline{P}_i,\mu) - \mu \log \epsilon_i \right \rceil,
\ee
where $\kappa (\overline{P}_i,\mu)$ is set to be equal to $\kappa$ in \Tref{thm-main2}, specified in \eq{kappa-def}, while setting $P_e = \overline{P}_i$ and using the same value for $\mu$. Note that $\kappa$ in \Tref{thm-main2} depends only on $P_e$ and $\mu$, and does not depend on $N$ and the sequence of channels $\left\{W_i\right\}_{i=1}^{N}$. Hence, $n_i$ can be set as specified in \eq{ni-def} and then by \Tref{thm-main2} we construct a polar code $\cC_i$ of length $N_i = 2^{n_i}$ and rate $R_i$ such that
\be{thm-main31}
R_i \geq  \overline{I}_{N_i}  - \kappa(\overline{P}_i,\mu)^{1/\mu} N_i ^{-1/\mu} \geq \overline{I}_{N_i}  - \epsilon_i,
\ee
where
$$
\overline{I}_{N_i} = \frac{1}{N_i}\sum_{j=1}^{N_i} I(W_j),
$$
and the probability of block error $P_i$ of $\cC_i$ under SC decoder when transmitted over $\left\{W_j\right\}_{j=1}^{N_i}$ is bounded by $\overline{P}_i$. Note that the left inequality in \eq{thm-main31} is by rearranging \eq{eq-main21} and the right inequality in \eq{thm-main31} is by the specific choice of $n_i = \log N_i$ in \eq{ni-def}. Note also that 
$$
\lim_{i \rightarrow \infty} \overline{I}_{N_i} = \overline{I}(\left\{W_i\right\}_{i=1}^{\infty}),
$$
and 
$$
\lim_{i \rightarrow \infty} (\overline{I}_{N_i} - \epsilon_i) = \overline{I}(\left\{W_i\right\}_{i=1}^{\infty}) - \lim_{i \rightarrow \infty} \epsilon_i = \overline{I}(\left\{W_i\right\}_{i=1}^{\infty}).
$$
Therefore, by \eq{thm-main31} and since $R_i < \overline{I}_{N_i}$ we have 
$$
\lim_{i \rightarrow \infty} R_i = \overline{I}(\left\{W_i\right\}_{i=1}^{\infty}).
$$
Also, since $P_i \leq \overline{P}_i$, we get
$$
\lim_{i \rightarrow \infty} P_i = 0,
$$
which completes the proof. 
\end{proof}

\section{Numerical Evaluation of the Bounds and Simulation Results}
\label{sec:new} 

In this section we discuss the numerical evaluation of the bounds derived on the speed of polarization and the scaling exponent of constructed polar codes. Also, simulation results are provided for non-stationary binary-input additive white Gaussian noise (BAWGN) channels. 

Recall the fine quantization described in Definition\,\ref{quant-def}. Note that the total number of sub-intervals in this quantization is $O(\log N)$ as long as $c$, $\lambda$, and $\tau$ are constant positive numbers. More specifically, the total number of sub-intervals is $c_1\log N + c_2$, where $c_1$ and $c_2$ depend on $c$, $\lambda$, and $\tau$. 

Let the function $h^*(z)$ for $z \in (0,1)$ be defined as follows:
\be{hstar-def}
h^*(z)\,\deff\, g(z,z),
\ee
where $g$ is defined in \eq{g-def}. And let 
\be{etastar-def}
\eta^*\,\deff\, -\log \sup_{z \in (0,1)} h^*(z).
\ee
Note that the value of $\eta$, defined in \eq{eta-def}, depends on the specific choice of $c$ and $\lambda$ which govern the definition of $h(.)$ in \eq{h-def}. In the next lemma, it is shown that $\eta \rightarrow \eta^*$, as $c,\lambda \rightarrow 0$, under certain conditions that are numerically verified later. 

\begin{lemma}
\label{etalem}
Let $\eta$ and $\eta^*$ be as defined in \eq{eta-def} and  \eq{etastar-def}, respectively. Suppose that the following three conditions hold:
\begin{enumerate}
\item $\eta^* = -\log h^*(z^*)$ for some $z^* \in (0,1)$,
\item $\lim_{z \rightarrow 0} \sup_{z' \in [z,2z]} g(z,z') < 2^{-\eta^*}$,
\item $\lim_{z \rightarrow 1}  \sup_{z' \in [z,(1+z)/2]} g(z,z') < 2^{-\eta^*}$,
\end{enumerate}
where $g$ is defined in \eq{g-def}. Then we have
\be{etalim}
\lim_{c,\lambda \rightarrow 0} \eta = \eta^*. 
\ee
\end{lemma}
\begin{proof}
Note that $g(z_1,z_2)$ is a continuous function with respect to both $z_1$ and $z_2$. Hence, for any $z \in [c,1-c]$, 
\be{etalem1}
\lim_{\lambda \rightarrow 0} \sup_{z' \in [z,\min(z+\lambda,1-c)]} g(z,z') = g(z,z) = h^*(z), 
\ee
where $h^*(.)$ is defined in \eq{hstar-def}. Given the three conditions in the statement of the lemma, we have $c < z^* < 1-c$ for small enough $c$ and
\be{etalem2}
\sup_{z' \in [z,2z]} g(z,z') < 2^{-\eta^*}
\ee
for $z \in (0,c)$, and
\be{etalem3}
 \sup_{z' \in [z,(1+z)/2]} g(z,z') < 2^{-\eta^*}
\ee
for $z\in (1-c,1)$. By definition of $h(z)$ in \eq{h-def} for $z \in (0,c)$ and $z \in (1-c,1)$ together with \eq{etalem2} and \eq{etalem3} we have
\be{etalem4}
\sup_{z \in (0,c) \cup (1-c,1)} h(z) < 2^{-\eta^*}. 
\ee
Also, by \eq{etalem1} and definition of $h(.)$ in \eq{h-def} for $z\in [c,1-c]$ we have
$$
\lim_{\lambda \rightarrow 0} h(z) = g(z,z) = h^*(z) \leq h^*(z^*) = 2^{-\eta^*}. 
$$
Also, the equality occurs for $z = z^*$. This together with \eq{etalem4} and the definition of $\eta$ in \eq{eta-def} complete the proof. 
\end{proof}

Note that the result of \Tref{thm-main} is still valid if we set very small values for $c$ and $\lambda$. In fact, the value of the constant $c_{\rho}$ depends also on $c$ and $\lambda$ but these parameters are treated as constants in the statement of \Tref{thm-main}. In stationary scenarios, $\eta^*$ is shown to be a bound on the speed of polarization (see \cite[Section IV]{hassani}). However, roughly speaking, \Tref{thm-main} together with \Lref{etalem} suggest the bound $\frac{\eta^*}{\eta^*+1}$ on the average speed of polarization $\overline{\eta}_n$ defined in \eq{speed-def2} for non-stationary scenarios. This will be elaborated more in the next subsection. The difference between the bounds in stationary and non-stationary scenarios is mainly due to the following, rigorously shown in the proof of \Tref{thm-main}: the number of pairs of synthetic channels whose $\Delta_f(.,.)$, defined in \eq{polar-rate}, is not bounded away from $1$ is not a constant, in a worst-case scenario, as the polarization process evolves. While this may seem an artifact of the specific proof technique used in \Tref{thm-main}, numerical examples for non-stationary BECs, in the next subsection, suggest that the bound of  \Tref{thm-main} may be very tight in these scenarios provided that the channels are initially sorted.

\subsection{Numerical results for non-stationary BECs}

\label{sec:newA} 

Consider a non-stationary sequence of channels $\{W_i\}_{i=1}^N$, where $N=2^n$ and all $W_i$'s are BECs. Also, suppose that the \textit{deterministic polarization 1}, defined in Definition\,\ref{polar1}, is applied to this sequence. It is shown in the following lemma that in order to ensure that the erasure probabilities of the synthetic channels within each sub-block of length $2^{n-j}$ in the $j$-th polarization level are in a non-increasing order, it is sufficient to initially sort $W_i$'s before applying the polarization transform. 
\begin{lemma}
\label{permbec}
Suppose that the erasure probabilities of BECs $W_i = \text{BEC}(p_i)$, where $p_i \in (0,1)$, are in a non-increasing order, i.e., $p_1 \geq p_2 ... \geq p_N$. Also, suppose that the polarization transform, as described in \eq{ch-comb} and without any modification, is applied to this sequence of channels. Then in each polarization level $j$, for $j \in [\![n]\!]$, the erasure probabilities of synthetic channels $W_{j,i}$, for $ i = (l-1)2^{n-j}+1, (l-1)2^{n-j}+2,\dots, l2^{n-j}$ for any $l \in  [\![2^j]\!] $, are in a non-increasing order. 
\end{lemma}
\begin{proof}
Note that if $1 > p_1 \geq p_2 \geq p_3 \geq p_4 > 0$, then we have $p_1 p_2 \geq p_3 p_4$
and
\begin{align*}
p_1 + p_2 - p_1 p_2 &= 1 - (1-p_1)(1-p_2)\\
& \geq 1 - (1-p_3)(1-p_4) = p_3 + p_4 - p_3 p_4.
\end{align*}
These together with the equations for combining operations of BECs described in Section\,\ref{sec:two} imply that the erasure probabilities of $W_{1,i}$'s, for $i \in  [\![N/2]\!]$, are in a non-increasing order. Similarly, the erasure probabilities of $W_{1,i}$'s, for $i \in  [\![N]\!] \setminus [\![N/2]\!]$, are also in non-increasing order. The proof is then complete by induction on $j$ and the recursive nature of polarization transform, described by \eq{ch-comb}. 
\end{proof}

In the case of non-stationary BECs, by using the equations for combining operations of BECs described in Section\,\ref{sec:two}, the definition of $g(.,.)$ in \eq{g-def} is modified as follows:
\be{g-def2}
g(z_1,z_2)\,\deff\, \frac{f(z_1z_2)+f(z_1 + z_2 - z_1 z_2)}{f(z_1)+f(z_2)},
\ee
where $f(.)$ is defined in \eq{f-def} for a parameter $b$ that we are going to select. Since $f(.)$ is a concave function, for any $z_1,z_2 \in (0,1)$, we have $g(z_1,z_2) < 1$. Hence, the condition specified in \eq{comb-cond} is always satisfied. 

Now, by following \eq{hstar-def} we have
\be{hstar-def2}
h^*(z) = \frac{f(z^2)+f(2z - z^2)}{2f(z)}.
\ee
Then the parameter $b$, that is used to specify $f$(.) in \eq{f-def}, is picked such that $\eta^*$, defined in \eq{etastar-def}, is maximized. More specifically, $b = \frac{2}{3}$ is picked, which is similar to what is picked for stationary BECs studied in \cite{hassani}. With $b=\frac{2}{3}$ the value of $\eta^*$ can be numerically estimated as follows:
\be{etastar-BEC}
\eta^* \approx 0.2669.
\ee
This closely matches the value provided in \cite[Table III]{hassani} corresponding to the value $m=0$ and the small discrepancy could be due to different numerical estimations. 

The conditions in \Lref{etalem} can be numerically verified. In particular, for $z^*= 0.158, 0.842$ we have $h^*(z^*) \approx 2^{-\eta^*}$, where $\eta^*$ is given by \eq{etastar-BEC}. Also, it can be observed that 
\begin{align*}
\lim_{z \rightarrow 0} \sup_{z' \in [z,2z]} g(z,z') &=  \lim_{z \rightarrow 0} g(z,2z) \\ 
&= \frac{3^b}{1+2^b} \approx 0.8039 \\
& < 2^{-\eta^*} \approx 0.8311.
\end{align*}
Also, similarly, we have   
$$
\lim_{z \rightarrow 1}  \sup_{z' \in [z,(1+z)/2]} g(z,z') \approx 0.8039 < 2^{-\eta^*}.
$$
Hence, the three conditions in \Lref{etalem} hold. Consequently, \Lref{etalem} implies that for any $\epsilon > 0$, one can pick small enough parameters $c$ and $\lambda$ in the fine quantization, defined in Definition\,\ref{quant-def}, such that we have 
$$
\eta > \eta^* - \epsilon,
$$
where $\eta$ is defined in \eq{eta-def} corresponding to the function $h(.)$ defined in \eq{h-def} given the particular $g(.,.)$ in \eq{g-def2} and $\eta^* \approx 0.2669$. Hence, by \Tref{thm-main}, we have the numerical lower bound 
$$
\overline{\eta}_n > \frac{\eta^*}{\eta^*+1} - \epsilon,
$$
on the average speed of polarization $\overline{\eta}_n$, defined in \eq{speed-def2}, as $n$ grows large. For instance, $0.2106$ is a numerical lower bound on the average speed of polarization. 

In order to see how tight the lower bound is, a numerical example is considered. Let $N = 2^{20}$, $n = 20$, and the erasure probabilities of $\{W_i\}_{i=1}^N$, already in a decreasing order, are elements of an arithmetic sequence starting at $0.99$ with the common difference $-0.98/N$, i.e., the last element and the average capacity are roughly $0.01$ and $0.5$, respectively. As mentioned earlier, the condition specified in \eq{comb-cond} is always satisfied. Also, the parameter $c$ can be picked arbitrarily small such that \eq{comb-cond1} is always satisfied in this numerical example. Hence, no operation is skipped in the \textit{deterministic polarization} 2, defined in Definition\,\ref{polar2}, applied to this example. Also, by \Lref{permbec} and since the erasure probabilities are initially in a decreasing order, there is no need for further sorting the bit-channels at each polarization level. Consequently, in this numerical example, \textit{deterministic polarization} 2 reduces to Ar{\i}kan's polarization transform without any modifications. 

\begin{figure}
\centering
\includegraphics[width=\linewidth]{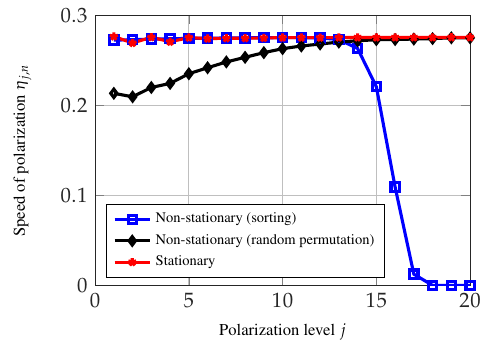}
\caption{Speed of polarization at level $j$, defined in \eq{speed-def1}, as a function of $j$}
\label{plot1}
\end{figure} 

The speed of polarization at level $j$ of the polarization, denoted by $\eta_{j,n}$ and defined in \eq{speed-def1}, is plotted in Figure\,\ref{plot1} as a function of $j$, for $j=1,2,\dots,20$, for the following three scenarios: for the given non-stationary sequence with erasure probabilities being already in a decreasing order (the first scenario), the same non-stationary sequence which is randomly permuted before the polarization transform is applied (the second scenario), and a stationary scenario where the erasure probability of all $N$ channels is $0.5$ (the stationary scenario). It can be observed that the speed of polarization for the first scenario initially follows that of the stationary scenario, perhaps thanks to the decreasing order, but at $n=14$ it starts to have a sharp transition to $0$. This actually aligns with how the proof of \Tref{thm-main} works, where the average speed of polarization at stage $k = \left\lceil n/(\eta+1)\right\rceil$ is bounded and then we simply say that $E_{n,n} \leq E_{k,n}$, where $E_{j,n}$'s are defined in \eq{En-def}. The reason for this step in the proof of \Tref{thm-main}, as opposed to following how $E_{j,n}$ evolves for $j=k+1,\dots,n$, is that for $j > k$ we have so many pairs of channels belonging to different sub-intervals, in the worst-case scenario, that polarization beyond level $k$ is simply discarded in the proof. More specifically, if one follows \eq{eq-thm5}-\eq{eq-thm2} for $k > \left\lceil n/(\eta+1)\right\rceil$, then the term $s2^{-n+k}$ in \eq{eq-thm2} becomes more significant compared with $2^{-\rho n}$, which adversely affects the resulting bound on the average speed of polarization. Hence, we instead invoke \eq{eq-thm2} for $k = \left\lceil n/(\eta+1)\right\rceil$ and then upper bound $E_{n,n}$ by $E_{k,n}$. The numerical example in this section suggests that upper bounding $E_{n,n}$ by $E_{k,n}$ may not be loose even though channel combining operations are not skipped in the polarization process. In fact, $\left\lceil n/(\eta+1)\right\rceil = 16$ in our example and $j=16$ falls into the sharp transition of the speed of polarization to $0$. In the second scenario, the speed of polarization is initially below that of the first and the stationary scenario, as expected. However, as $j$ increases, it reaches that of the stationary scenario. This can be interpreted as follows. Since a random permutation is applied, as $j$ increases, the synthetic channels within sub-blocks of length $2^{n-j}$ are becoming more and more \textit{stationary-like}, i.e., there is smaller and smaller variation in their erasure probabilities. Hence, the speed of polarization follows that of the stationary scenario for larger $j$'s. 

Also, in this numerical example, the average of $\eta_{j,n}$'s, for $j=1,2,\dots,n$ where $n=20$, is $0.2087$ in the first scenario, which is, interestingly, very close to the numerical lower bound $0.2106$ that we have proved on the average speed of polarization as $n$ grows large. This suggests that our approach may lead to numerically tight lower bounds on the average speed of polarization for non-stationary BECs assuming that the BECs are initially sorted. The average of $\eta_{j,n}$'s in the stationary scenario is $0.2749$. As expected, this is slightly less than $0.2757 = 1/3.627$, where $3.627$ is the parameter for stationary BECs that is numerically approximated in \cite{hassani}. The average of $\eta_{j,n}$'s in the second scenario is $0.2545$, which is larger than that of the first scenario but smaller than that of the stationary scenario. 

In general, it is expected that, as shown in the numerical results for the second scenario, with a random permutation the speed of polarization $\eta_{j,n}$ converges to that of the stationary scenario as $j$, the polarization level, increases up to $n$. However, the average speed of polarization might be still lower than that of the stationary scenario. Also, a scheme with a random permutation does not immediately lead to explicit constructions, i.e., one has to pick a permutation from a random ensemble for theoretical guarantees. For practical purposes, it is an interesting problem to study how to optimize the selection of the permutation that is initially applied to the non-stationary sequence of channels. This problem has been studied for certain numerical examples in \cite{zorgui2019non}. 

Regarding the finite-length scaling exponent of the capacity-achieving codes, constructed in Section,\ref{sec:five}, \eq{mu-def} implies that 
$$
\mu > 2+\log 3 + \frac{1}{\eta^*} 
$$
is sufficient for \Tref{thm-main2} to hold, where $\eta^*$ is given in \eq{etastar-BEC}. This suggests the numerical upper bound of $7.34$ on the scaling exponent of the constructed capacity-achieving codes for non-stationary BECs.  

\subsection{Numerical results for general non-stationary BMS channels}
\label{sec:newB}

Consider a non-stationary sequence of BMS channels $\{W_i\}_{i=1}^N$, where $N=2^n$. Also, suppose that the \textit{deterministic polarization 2}, defined in Definition\,\ref{polar2}, is applied to this sequence. \textit{Proper} permutations, as discussed in Section\,\ref{sec:three}, are applied at each polarization level to ensure that the Bhattacharyya parameters of synthetic channels are in a non-increasing order within sub-blocks of length $2^{n-j}$ at polarization level $j$. Also, channel combining operations are performed if \eq{comb-cond1} and \eq{comb-cond} are both satisfied. 

In light of \Lref{etalem}, in order to numerically estimate a bound on the average speed of polarization one needs to numerically compute $\eta^*$ in \eq{etastar-def}. A larger $\eta^*$ would imply a better bound on the average speed of polarization. Hence, the parameter $b$ in the definition of $f(.)$ in \eq{f-def} needs to be optimized. Note that one can do better by adopting the strategy proposed in \cite{hassani}. More specifically, the function
\be{fnew-def}
f(z) = \frac{1}{20}(8z^2+5z+19)z^{\frac{3}{4}}(1-z)^{\frac{3}{4}}
\ee
is chosen, which leads to $\eta^* \approx 0.202$ \cite[Section IV-B]{hassani}. 

The conditions in \Lref{etalem} can be numerically verified. In particular, for $z^* = 0.178$ we have $h^*(z^*) \approx 2^{-\eta^*}$, where $\eta^* \approx 0.202$. Also, it can be observed that 
\begin{align*}
\lim_{z \rightarrow 0} \sup_{z' \in [z,2z]} g(z,z') &=  \lim_{z \rightarrow 0} g(z,2z) \\ 
&= \frac{3^{3/4}}{1+2^{3/4}} \approx 0.85 \\
& < 2^{-\eta^*} \approx 0.869.
\end{align*}
Also, similarly, we have   
$$
\lim_{z \rightarrow 1}  \sup_{z' \in [z,(1+z)/2]} g(z,z') \approx 0.85 < 2^{-\eta^*}.
$$
Hence, the three conditions in \Lref{etalem} hold. Consequently, \Lref{etalem} implies that for any $\epsilon > 0$, one can pick small enough parameters $c$ and $\lambda$ in the fine quantization, defined in Definition\,\ref{quant-def}, such that we have 
$$
\eta > \eta^* - \epsilon.
$$
This value of $\eta$ can be then used in \Tref{thm-main}, \Tref{thm-main2}, and \Tref{thm-main3}. One small difference, comparing to the case when a function $f(.)$ of the type specified in \eq{f-def} is selected, is that the value of $c_{\rho}$ in \Tref{thm-main} is slightly increased. More specifically, with the choice of $f(.)$ in \eq{fnew-def}, we have
\be{fnew-cond} 
f(z) < \frac{8}{5}z^{\frac{3}{4}}(1-z)^{\frac{3}{4}},
\ee
for $z \in (0,1)$. Then in the proof of \Tref{thm-main}, $b = \frac{3}{4}$ is used. Also, the right-hand side of \eq{dj-cond} is multiplied by a factor of $\frac{8}{5}$. This together with \eq{fnew-cond} ensure that \eq{dj-cond} still holds. Consequently, \eq{eq-thm2} is slightly modified and then the argument of $\log(.)$ in \eq{crho-def} is increased by an additional factor of $\frac{3}{5}$, which implies that $c_{\rho}$ is slightly increased.  

To summarize, \eq{mu-def} implies that 
\be{mu-bound}
\mu > 2+\log 3 + \frac{1}{\eta^*} 
\ee
is sufficient for \Tref{thm-main2} to hold. Hence, we obtain $8.54$ as a numerical upper bound on the scaling exponent of polar codes constructed for general non-stationary BMS channels.

Note that exactly characterizing the scaling exponent of polar codes, constructed from the $2 \times 2$ kernel, for general BMS channels is still open in the stationary case. In fact, it is not known whether the methods introduced in \cite{hassani}, and developed further in \cite{GB,MUH}, to study the speed of polarization and the scaling exponent of the constructed polar codes, which we also adopt here to some extent, are adequate to precisely characterize these parameters for general BMS channels or not. 

\subsection{Simulation results for non-stationary BAWGN channels}
\label{sec:newC} 

In this section Monte-Carlo simulation results for non-stationary BAWGN channels are provided and the results are compared with stationary BAWGNs. 

The input to BAWGN$(\sigma)$ is $x \in \{-1,1\}$, and the output is $x+n$, where $n \sim \cN(0,\sigma^2)$.  By convention, the value of signal-to-noise ratio (SNR) in dB scale for a BAWGN$(\sigma)$ is defined as $10 \log _{10} 2\sigma^2$. The sequences of non-stationary BAWGN channels to be simulated in this section are defined with respect to their SNRs, where SNRs follow an arithmetic sequence. Let $N=2^{10}$, where $N$ is the number of channels in the sequence and is also equal to the code block length. The common difference in each considered arithmetic sequence is $1/N$. For instance, one considered sequence of non-stationary BAWGNs is given by the following SNRs: $\{-2+\frac{i}{N}\}_{i=1}^N$ in dB scale. Each such sequence is then represented by its \textit{effective average} SNR, defined to be equal to the SNR of the BAWGN channel whose capacity is equal to the average of the capacities of the BAWGN channels in the given sequence (this, in fact, is very close to the middle term of the considered arithmetic sequences, e.g., the  \textit{effective average} SNR of the sequence of SNRs $\{-2+\frac{i}{N}\}_{i=1}^N$ in dB scale is approximately $-1.5$\,dB). This provides a fair way of comparing the performance of polar codes in non-stationary and stationary scenarios when plotting the performance, in terms of the block error rate versus the SNR. For ease of description, the \textit{effective average} SNR in non-stationary scenarios and SNR in stationary scenarios are both simply referred to as SNR. 

\begin{figure}
\centering
\includegraphics[width=\linewidth]{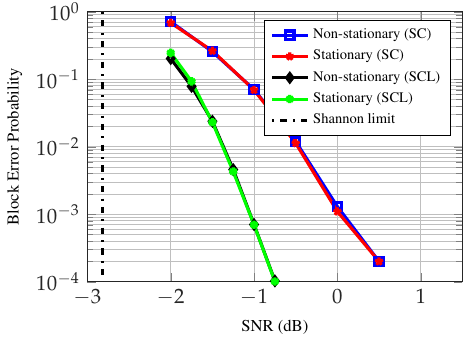}
\caption{Performance comparison at the block length $N=1024$ and the rate $\shalf$. The SCL is with maximum list size of $16$ and a CRC of length $16$. The Shannon limit is $-2.823$\,dB, i.e., the capacity of a BAWGN channel with SNR $=-2.823$\,dB is $\shalf$.}
\label{plot2}
\end{figure} 

In order to construct polar codes, Tal-Vardy construction method is considered \cite{TV2}. In particular, the method proposed in \cite{TV2} based on degrading quantization, while combining synthetic channels, is adopted in both stationary and non-stationary scenarios. The main purpose of degrading quantization is to ensure a bounded output alphabet size for synthetic channels \cite{TV2}. Note that degrading quantization is an operation that is done on a single BMS channel in order to reduce its output alphabet size while limiting the capacity loss. Hence, it can be naturally adopted in the non-stationary polarization transform as well. The quantization parameter is set to $16$. Also, the continuous output alphabet of BAWGNs is initially quantized, using the method proposed in \cite[Section VI]{TV2}, and BMS channels with discrete output alphabet of size $10^3$ are obtained. Polar codes of rate $\shalf$ are constructed in both stationary and non-stationary scenarios at SNR$=-1$\,dB, where the considered non-stationary sequence corresponds to an arithmetic sequence, as described before. The total average capacity loss after $10$ levels of polarization, as a result of channel quantization and degrading quantization, is less than $0.002$. For the non-stationary scenario the \textit{deterministic polarization 2}, defined in Definition\,\ref{polar2}, is adopted. The condition in \eq{comb-cond}, with $f(.)$ given by \eq{fnew-def}, is slightly relaxed by adding $10^{-4}$ to the right hand side of \eq{comb-cond} to compensate the effect of quantization degradation. It is observed that this condition is always satisfied for all pairs of the synthetic channels in all $10$ levels of polarization. Also, $c$ can be picked arbitrarily small in \eq{comb-cond1}. As a result, no channel combining operation is skipped. Also, although the channels are initially sorted, permutations need to be applied in each polarization level, as opposed to the non-stationary BEC case discussed in \Lref{permbec}, to ensure that the Bhattacharyya parameters in each sub-block of length $2^{n-j}$ at level $j$ are in a non-increasing order. 

The performances of the constructed polar codes of rate $\shalf$, as described above, are evaluated over a range of SNRs and the results are shown in Figure\,\ref{plot2}. Two types of decoders are considered: the successive cancellation (SC) method proposed by Ar{\i}kan \cite{Arikan} and the successive cancellation list (SCL) decoding with the help of cyclic redundancy check (CRC) bits proposed in \cite{TV}. A maximum list size of $16$ is assumed and a CRC of length $16$ is considered in the SCL decoding. In order to have an effective code rate of $\shalf$, $512+16 = 528$ best synthetic channels, in terms of their degraded probability of error estimated by Tal-Vardy construction method, at the last polarization level are selected to carry the information bits. It is observed that the performance in the two cases are very close with a negligible difference. 

The simulation results in this section may seem to oppose the conclusion in section\,\ref{sec:newA} about slower speed of polarization for non-stationary BECs comparing to stationary BECs. However, we argue that this might be due to various reasons. The variation in the capacity of the non-stationary channels considered in section\,\ref{sec:newA} is from $0.01$ to $0.99$. Roughly speaking, this can be considered as an \textit{extreme} case in the non-stationary scenario. Intuitively, it is expected that with more variation in the capacity of channels in the non-stationary sequence, a smaller speed of polarization is obtained. The variation in the capacity of BAWGN channels in this section, for instance, for the channels with SNRs $\{-2+\frac{i}{N}\}_{i=1}^N$ in dB scale is from $0.564$ to $0.643$. In order to have a variation in the capacity from $0$ to $1$, as in BECs, for BAWGN channels, one has to assume SNRs varying from $-\infty$ to $\infty$, a situation which does not seem practical. Another reason is that a very large $N$, i.e., $N=2^{20}$, is considered in section\,\ref{sec:newA} in order to illustrate the difference in the speed of polarization between the stationary and the non-stationary scenarios. However, the simulation results in this section suggest that at moderate block lengths of practical interest, e.g., $N=1024$, and with a \textit{reasonable} variation in the capacity of channels in the non-stationary sequence, polar codes with the same rate and a very similar performance, comparing to polar codes for stationary scenarios, can be constructed.

\section{Discussions and Conclusion}
\label{sec:six}
In this paper we considered the problem of polar coding for non-stationary sequences of channels where the channels are independent but not identical. Several modifications to Ar{\i}kan's channel polarization transform are proposed. The end result is that an arbitrary sequence of non-stationary channels can be polarized while lower bounding the speed of polarization. Also, the resulting fast polarization scheme is combined with another polarization stage in order to construct polar coding schemes that achieve the average capacity of non-stationary channels. Furthermore, it is shown that the block length $N$ of the constructed code is upper bounded by a polynomial of the inverse of the gap to the average capacity. Several numerical and simulation results are provided at finite block lengths. 

There are several directions for future research. Improving the upper bound on the finite-length scaling of the constructed polar codes is an interesting problem. For instance, if one can avoid the skipped channel combining operations using an alternative technique, that can immediately improve the upper bound. Also, it remains open whether any of the modifications proposed in this paper, i.e., sorting the bit-channels and skipping operations, are actually necessary in order to construct explicit capacity-achieving polar codes with bounded scaling exponent. In other words, the question is the following: Does there exist a non-stationary sequence of channels for which fast polarization does not occur when applying Ar{\i}kan's polarization transform without any modification? Furthermore, studying the error exponent of polar codes over non-stationary channels, as an extension of \cite{AT} to the non-stationary case, and unified scaling laws for non-stationary channels, as an extension of the results of \cite{MUH} in the stationary case, are other interesting problems.

\bibliographystyle{IEEEtran}
\bibliography{polar}

\section*{Appendix}

\begin{lemma}
\label{Z--lemma}
For any two BMS channels $W_1$ and $W_2$, we have
\be{eqZ3}
Z(W_1 \boxcoasterisk W_2) \geq \sqrt{z_1^2 + z_2^2 - z_1^2 z_2^2},
\ee
where $z_1 = Z(W_1)$ and $z_2 = Z(W_2)$.
\end{lemma}
\begin{proof}
It can be shown that \eq{eqZ3} turns to equality when both $W_1$ and $W_2$ are binary symmetric channels (BSC). It is well-known that any BMS channels can be decomposed into a convex combination of multiple BSCs. Also, note that $\sqrt{z_1^2 + z_2^2 - z_1^2 z_2^2}$ is convex in both $z_1$ and $z_2$ when the other parameter is fixed. Therefore, the lemma follows.
\end{proof}

\begin{lemma}
\label{IZ-lemma}
For any BMS channel $W$, we have
$$
I(W) \leq 1 - Z(W)^2.
$$
\end{lemma}
\begin{proof}
We use the fact any BMS channel can be decomposed into a convex combination of multiple BSCs, same as in the proof of \Lref{Z--lemma}. Since $-z^2$ is a concave function, it is then sufficient to prove the lemma for a BSC $W$. If $W$ is BSC$(p)$, then we have
$$
I(W) = 1+p\log p + (1-p)\log (1-p),
$$
and
$$
Z(W) = 2\sqrt{p(1-p)}.
$$
Therefore, it suffices to show that
\be{IZ-lemma1}
g(p) \,\deff\, -p\log p - (1-p)\log (1-p) -4p(1-p) \geq 0. 
\ee
Since $g(p) = g(1-p)$, it is enough to show \eq{IZ-lemma1} holds for $p \in (0,0.5]$ (for $p=0$, $g(p) = 0$). Note that
$$
g'(p) = \log \frac{1-p}{p} +8p-4,
$$
and
$$
g''(p) = -\frac{\log e}{p(1-p)} + 8,
$$
where $g'$ and $g''$ are also continuous over $(0,0.5]$. Since $g''(p)$ is increasing over $(0,0.5]$, it is zero for at most one $p \in (0,0.5]$ and hence, $g'(p) = 0$ has at most two solutions over $(0,0.5]$ which are
$$
p_1 \approx 0.096,\ \text{and} \ p_2 = 0.5.
$$
Note that $g(p)$ is increasing for $p \in [0,p_1]$ and $g(p_2) = 0$. Therefore, $\min_{p \in [0,0.5]} g(p) = 0$ and the lemma follows.
\end{proof}

\begin{lemma}
\label{IZ-lemma2}
Let $b,t \in (0,1)$, and $f(z) = z^b(1-z)^b$. Then for any sequence of BMS channels $\left\{W_i\right\}_{i=1}^N$
$$
\sum_{z_i > t} I(W_i) \leq \frac{2}{t}\sum_{i=1}^{N} f(z_i),
$$
where $z_i = Z(W_i)$,
\end{lemma}
\begin{proof}
Using \Lref{IZ-lemma} we have
\begin{align*}
\sum_{z_i > t} I(W_i) \leq \sum_{z_i > t} (1-z_i^2)
\leq \frac{2}{t} \sum_{z_i > t} z_i(1-z_i)
\leq \frac{2}{t} \sum_{z_i > t} f(z_i).
\end{align*}
\end{proof}
\begin{corollary}
\label{C-key}
Under the same assumptions in \Lref{IZ-lemma2}, 
$$
\Bigl|\left\{i \in [\![N]\!] : z_i \leq t \right\}\Bigr| \geq \sum_{i=1}^N I(W_i) - \frac{2}{t}\sum_{i=1}^{N} f(z_i).
$$
\end{corollary}
\begin{proof}
The proof is by noting that
$$
\Bigl|\left\{i \in [\![N]\!] : z_i \leq t \right\}\Bigr| \geq \sum_{z_i \leq t} I(W_i) =  \sum_{i=1}^N I(W_i) - \sum_{z_i > t} I(W_i)
$$
followed by using \Lref{IZ-lemma2}. 
\end{proof}

\begin{lemma}
\label{partition-lemma}
Let $N=MK$, where $N,M,K \in \N$, and let $A = \left\{x_1,x_2,\dots,x_N\right\}$, where $x_i \in [0,1]$, for $i \in [\![N]\!]$, and
$$
\lambda = \frac{1}{N} \sum_{i=1}^N x_i.
$$
Then one can partition $A$ into $K$ subsets $A_1,A_2,\dots,A_K$ of size $M$ such that
\be{part1}
 \min_{j \in [\![K]\!]} \lambda_j \geq \lambda - \frac{1}{M},
\ee
where $\lambda_j$ is the average of the elements in $A_j$. 
\end{lemma}
\begin{proof}
Since the number of possible partitions is finite, one can find the partition such that 
$$
\min_{j \in [\![K]\!]} \lambda_j
$$
is maximum among all the possible partitions. We show that this partition must satisfy \eq{part1}. Let 
$$
j_0 = \text{argmin}_{j \in [\![K]\!]} \lambda_j.
$$
Assume to the contrary that \eq{part1} does not hold, i.e.,
$$
\lambda_{j_0} < \lambda - \frac{1}{M}.
$$
Note that there exists $j_1 \in [\![K]\!]$ such that
$$
\lambda_{j_1} > \lambda.
$$
Let $x_{i_0}$ be the minimum element in $A_{j_0}$ and $x_{i_1}$ be the maximum element in $A_{j_1}$. We make a new partition by swapping the elements $x_{i_0}$ and $x_{i_1}$ in $A_{j_0}$ and $A_{j_1}$. Note that 
$$
0 < x_{i_1} - x_{i_0} \leq 1.
$$
Therefore, the average of $A_{j_0}$ is strictly increased in the new partition, while the average of $A_{j_1}$ is decreased by at most $1/M$, which means it is still greater than the old average of $A_{j_0}$. Hence, 
$$
\min_{j \in [\![K]\!]} \lambda_j
$$ 
is strictly increased in the new partition which is a contradiction. This proves the lemma.
\end{proof}

\begin{lemma}
\label{2power-lemma}
Let $N = 2^n$, $2 \leq j \leq n/(1+\log 3)$, $\alpha,\beta > 0$ and $a_1,a_2,\dots,a_N \in \N$ such that $\sum_{i=1}^N a_i \leq 3^j$. Then we have
\be{2power-1}
\begin{split}
&\log \Bigl|\left\{i \in [\![N]\!]: 2^{s_j(i)} - a_i \leq \alpha n + \beta\right\}\Bigr| \\
&\leq n-j+ p(n,\alpha,\beta),
\end{split}
\ee 
where $s_j(i)$ is defined in \eq{sj-def} and
\be{p-def}
\begin{split}
p(n,\alpha,\beta)\deff &\bigl(2 - \log \gamma + \log n + \log (\alpha n + \beta)\bigr) (\log n - \log \gamma) \\
&+ \log \bigl(3 - \log \gamma + \log n + \log (\alpha n + \beta)\bigr),
\end{split}
\ee
where $\gamma = 1+\log 3$.
\end{lemma}
\begin{proof}
Let
$$
A = \left\{i \in [\![N]\!]: 2^{s_j(i)} - a_i \leq \alpha n + \beta \right\},
$$
and $r = |A|$. Note that $r \leq N = 2^n$. Let $t$ denote the smallest integer such that $t \geq 2$ and
\be{2power-0}
r \leq 2^{n-j} \sum_{l=0}^t {j \choose l}. 
\ee
Note that $t \leq j$. Then we have
\be{2power-2}
2^{n-j} \sum_{l=0}^{t-1} {j \choose l} 2^l \leq \sum_{i \in A} 2^{s_j(i)}.
\ee 
By definition of $A$ and \eq{2power-0} we have
\be{2power-3}
\sum_{i \in A} 2^{s_j(i)} \leq (\alpha n +\beta) r + \sum_{i=1}^N a_i \leq (\alpha n+\beta) 2^{n-j} \sum_{l=0}^t {j \choose l} +  3^j.
\ee 
By combining \eq{2power-2} and \eq{2power-3} we get
\be{2power-4}
2^{n-j}\sum_{l=0}^{t-1} {j \choose l} 2^l \leq (\alpha n+\beta) 2^{n-j} \sum_{l=0}^t {j \choose l} +  3^j.
\ee 
Since $j \leq n/(1+\log 3)$, we have $2^{n-j} \geq 3^j$. Combining this with \eq{2power-4} we get
\be{2power-5}
{j \choose t-1} 2^{t-1} \leq (\alpha n+\beta) \sum_{l=0}^t {j \choose l} \leq (\alpha n+\beta) (t+1) {j \choose t},
\ee
and therefore,
\be{2power-6}
2^{t-1} \leq 2 (\alpha n+\beta) j,
\ee
or equivalently
\be{2power-7}
\begin{split}
t &\leq 2 + \log (\alpha n+\beta) + \log j \\
& \leq 2 - \log(1+\log 3)+ \log n + \log (\alpha n + \beta).
\end{split}
\ee
By \eq{2power-0} and \eq{2power-7} we have
\begin{align*}
r \leq 2^{n-j} \sum_{l=0}^t {j \choose l} &\leq 2^{n-j} (t+1) j^t,
\end{align*}
and the lemma follows by taking logarithm of both sides and the bound on $t$ given in \eq{2power-7}.
\end{proof}

\begin{lemma}
\label{det-lemma1}
Given a sequence of channels $\left\{W_i\right\}_{i=1}^N$, where $N=2^n$, let $x_{0,i} = Z(W_i)$ and consider the extremal deterministic process $\left\{x_{j,i}\right\}$, for $j=0,1,\dots,n$ and $i \in [\![N]\!]$, defined in Definition\,\ref{ext-proc}. Let also $\left\{W_{j,i}\right\}$ denote the \textit{deterministic polarization} 2, defined in Definition\,\ref{polar2}, with the set of initial channels $W_{0,i} = W_i$ such that the permutations and the indices of skipped operations follow the extremal deterministic process $\left\{x_{j,i}\right\}$. Then we have
$$
Z(W_{j,i}) \leq x_{j,i},
$$
for $i \in [\![N]\!]$ and $j=0,1,\dots,n$.
\end{lemma}
\begin{proof}
The proof is by induction on $j$ and using \eq{eqZ1} and \eq{eqZ2} on the Bhattacharyya parameters of the combined channels and the definition of extremal deterministic process in Definition~\ref{ext-proc}. 
\end{proof}

\end{document}